\documentclass[10pt, leqno, oneside]{amsart}
\usepackage{amsfonts}
\usepackage{amsmath, amsthm, amssymb}
\usepackage{comment} % provides the {comment} environment
\reversemarginpar
\usepackage[dvips,left=1in, right=1in,
top=1in, bottom=1in, headsep=0.4in,footskip=0.4in]{geometry}
\numberwithin{equation}{section}
\theoremstyle{plain}                % title and number in bold, text italic
\newtheorem{theorem}{Theorem}[section]
\newtheorem{lemma}[theorem]{Lemma}
\newtheorem{proposition}[theorem]{Proposition}
\newtheorem{corollary}[theorem]{Corollary}
\theoremstyle{definition}           % title and number in bold, text normal
\newtheorem{definition}[theorem]{Definition}
\newtheorem{example}[theorem]{Example}

\newtheorem{assumption}[theorem]{Assumption}
\theoremstyle{remark}
\newtheorem{remark}[theorem]{Remark}

\newcommand{\X}[1]{\mathcal{X}(#1)}
\newcommand{\xinf}{\mathcal{X}^{\infty}}
\newcommand{\ar}{\mathcal{R}}
\newcommand{\ra}{\rho_{\mathcal{A}}}
\newcommand{\Ma}{\mathcal{M}_{\mathcal{A}}}
\newcommand{\bsp}{\mathbf{p}}
\newcommand{\hr}{\hat{\mathcal{R}}}
\newcommand{\boa}{\boldsymbol{a}}

\newcommand{\argmin}{\operatorname{argmin}}

\DeclareMathOperator*\epi{epi}
\newcommand{\abs}[1]{\left| #1 \right|} % absolute value
\newcommand{\set}[1]{\left\{#1\right\}} % curly brackets
\newcommand{\sets}[2]{\set{#1\,:\,#2}} % a set with "such that"
\newcommand{\ind}[1]{ {\mathbf 1}_{{#1}}} % indicator of a set
\newcommand{\norm}[1]{{||#1||}} % norm enclosure

\newcommand{\prf}[1]{ ( #1 )_{t\in [0,T]}}

\providecommand{\R}{} \renewcommand{\R}{{\mathbb R}}
 
\newcommand{\N}{{\mathbb N}}
\newcommand{\PP}{{\mathbb P}}
\newcommand{\QQ}{{\mathbb Q}}

\newcommand{\EE}{{\mathbb E}}
\newcommand{\FF}{{\mathcal F}}
\newcommand{\GG}{{\mathcal G}}
\newcommand{\BB}{{\mathcal B}}

\newcommand{\MM}{{\mathcal M}}

\newcommand{\CC}{{\mathcal C}}
\newcommand{\EN}{{\mathcal E}}
\newcommand{\XX}{{\mathcal X}}

\renewcommand{\AA}{{\mathcal A}}

\newcommand{\eps}{\varepsilon}
\newcommand{\ld}{\lambda}

\newcommand{\el}{{\mathbb L}} %l-pees
\newcommand{\lzer}{\el^0}
\newcommand{\lone}{\el^1}

\newcommand{\linf}{\el^{\infty}}
\newcommand{\remove}[1]{\st{#1}}
\renewcommand{\remove}[1]{}
\newcounter{notenum}

\newcommand{\define}[1]{{\em #1}}

\newcommand{\bS}{\mathbf{S}}
\newcommand{\bT}{\boldsymbol{\Theta}}
\newcommand{\bt}{\boldsymbol{\vartheta}}
\newcommand{\bvt}{\boldsymbol{\vartheta}}
\newcommand{\bB}{\boldsymbol{B}}

\newcommand{\agset}{\mathcal{G}_{\mathcal{E}_{1},\mathcal{E}_{2}}}

\renewcommand{\agset}{\mathcal{G}}

\newcommand{\tAA}{\tilde{\AA}}

\newcommand{\bsv}{\boldsymbol{v}}
\newcommand{\bsdel}{\boldsymbol{\delta}}
\newcommand{\sP}{\mathcal{P}}
\newcommand{\bze}{\boldsymbol{0}}

\newcommand{\lzr}{\rho_1\lozenge\dots \lozenge \rho_I}
\newcommand{\sek}[1]{\{#1_k\}_{k\in\N}}

\begin{document}

\title{Partial Equilibria with Convex Capital Requirements:\\
  Existence, Uniqueness and Stability}

\thanks{Both authors were
  supported in part by the National Science Foundation under award
  number DMS-0706947 during the preparation of this work. Any
  opinions, findings and conclusions or recommendations expressed in
  this material are those of the authors and do not necessarily
  reflect those of the National Science Foundation. \linebreak\indent
  The authors would like to thank Ioannis Karatzas, Constantinos
  Kardaras, Mihai Sirbu, Stathis Tompaidis, Thaleia Zariphopoulou, and
  the participants of the Fifth World Congress of the Bachelier
  Finance Society, London, UK, 2008, for fruitful discussions and good
  advice.}

\maketitle

\bigskip

\begin{center}
\begin{minipage}{0.4\textwidth}
\begin{center}
{\bf\large Michail Anthropelos}\\
Department of Mathematics\\
University of Texas at Austin\\
1 University Station, C1200\\
Austin, TX 78712, USA\\
{\tt manthropelos@math.utexas.edu}\\
\end{center}
\end{minipage}
\begin{minipage}{0.4\textwidth}
\begin{center}
{\bf\large Gordan \v Zitkovi\' c}\\
Department of Mathematics\\
University of Texas at Austin\\
1 University Station, C1200\\
Austin, TX  78712, USA\\
{\tt gordanz@math.utexas.edu}\\
\end{center}
\end{minipage}
\end{center}

\ \\[1ex]

\begin{center} \today \end{center}

\bigskip

\begin{quote}
  \noindent{\bf Abstract.} In an incomplete semimartingale model of a
  financial market, we consider several risk-averse financial agents
  who negotiate the price of a bundle of contingent claims.  Assuming
  that the agents' risk preferences are modelled by convex capital
  requirements, we define and analyze their demand functions and
  propose a notion of a partial equilibrium price. In addition to
  sufficient conditions for the existence and uniqueness, we also show
  that the equilibrium prices are stable with respect to
  misspecifications of agents' risk preferences.
 \end{quote}

\medskip


\medskip


\section{Introduction}

In complete market models, the price of a contingent claim is simply
given by its replication cost.  In the more realistic, incomplete
models, the arbitrage-free paradigm typically fails to produce a
unique price and yields only a price-interval.  The presence of
\textit{unhedgeable} claims - due to the aforementioned market
incompleteness - necessitates the introduction of another fundamental
principle whenever one wants to produce a unique value for a given
contingent claim. The long history of empirical inquiry into human
behavior under risk dictates that this additional component is related
to some numerical measure of risk-aversion, idiosyncratic to the
agent valuing the claim.  The majority of the existing literature uses
agents' risk preferences to induce a \textit{subjective} ``pricing''
mechanism which provides bid and ask prices for a claim payoff
(consider for instance the indifference- or marginal-utility-based
price concepts; see, e.g., the references in \cite{Car09}). In
reality, however, the observed price of any claim is always a result
of  interaction among a number of agents. In fact, the very notion
of a ``price'' makes sense only as the observed quantity at which a
transaction between two (or more) agents already took place;
consequently, what is called {\em pricing} in the bulk of the
contemporary literature should rather be referred to as {\em valuation}. We
abstain from such a renaming in order to keep in line with the already
established terminology.

In the present paper, we consider several risk-averse financial agents
who negotiate the price of a fixed bundle of claims, and we propose a
{\em partial-equilibrium} pricing scheme in the spirit of the
classical general-equilibrium theory.  We place ourselves in a
(liquid) financial market model driven by a \textit{locally-bounded
  semimartingale}, fix a time horizon $T>0$ and assume that each
agent's risk preferences and investment goals are abstracted in the
notion of an \textit{acceptance set}. Roughly speaking, this set
includes all the positions with maturity up to time $T$ that the
agent deems acceptable at time $0$. Following the literature of convex
risk measures, we assume that each acceptance set satisfies certain
standard properties, such as monotonicity and convexity. An additional
property that relates the agents' acceptance sets to the liquid market and
the set of admissible strategies is also imposed (see Axiom
\textit{Ax4} on page \pageref{axiom 4}).  Thus axiomatized acceptance
sets are naturally identified with \textit{capital requirements} (or
\textit{risk measures});
intuitively, the (convex) capital requirement of a payoff is the
minimum amount of money which, when added to the payoff, creates
a position in the  acceptance set.

The notion of risk measure was introduced to Mathematical Finance in
the late nineties (see \cite{ArtDelEbeHea99}) and has captured a large
part of the research activity in this field since (see, among others,
\cite{Del02} and \cite{FolSch02}, as well as Chapter 4 in
\cite{FolSch04} and references therein). Convex risk measures in the
context of a liquid financial market were first studied in
\cite{CarGemMad01} and \cite{FolSch02} (see also \cite{FolSch02b} and
section 4.8 of \cite{FolSch04}). In \cite{FolSch02}, a convex risk
measure is defined in the presence of a financial market, where the
agent is allowed to trade in discrete time and under specific convex
constrains. In \cite{CarGemMad01}, a pricing scheme for non-replicable
claims based on risk measures is proposed in a
finite-probability-space model. The abstract definition of a convex
capital requirement and its dual representation for a large family of
models was given in \cite{FriSca06} and \cite{BarElk05}.  For the
dynamic version of convex-risk-measure-based pricing in an incomplete
market setting, we refer the reader to \cite{KloSch07} and, for the
sufficient conditions for existence of an optimal trading strategy
that makes a contingent claim acceptable, to \cite{Xu06} and
\cite{Pal06}.

Having described the decision-theoretic set up, we focus on the
interaction of $I\geq 2$ agents, who have access to (possibly)
different financial markets, and we define \textit{mutually agreeable
  bundle}. Given a bundle of contingent claims
$\bB=(B_1,B_2,...,B_n)$, we start by introducing the set of its {\em
  allocations}, i.e., the set of matrices that represent the feasible
ways of sharing $\bB$ among agents. Then, we say that a pair
$(\bB,\boa)$, of a bundle $\bB$ and its allocation $\boa$, is {\em
  mutually agreeable} if there exists a price vector $\bsp$, at which
re-allocation of $\bB$ according to $\boa$ is acceptable to every
agent at price $\bsp$.  This is a generalized version of the notion of
mutually agreeable claims given in \cite{AntZit08}. In section
\ref{sec:agreement}, we study its properties and relate it to the
well-known notion of \textit{Pareto optimality}.

For models that include uncertainty, the concept of a Pareto optimal
allocation was first analyzed in the insurance/reinsurance context in
\cite{Borc62}, \cite{Ger78}, and \cite{BuhJew79} and further developed
in \cite{Buh80}, \cite{Buh84} and \cite{Wyl90}. More recently, the
issues related to Pareto optimality and design of an optimal contract
were studied in the more general settings of convex (coherent) risk
measures (see, e.g., \cite{BarElk04}, \cite{BarSca08},
\cite{BurRus07}, \cite{FilKup08}, \cite{HeaKu04} and
\cite{JouSchTou06}).  In the presence of a financial market, this
problem was addressed in \cite{BarElk05} and \cite{KloSch07}.
Recently, in \cite{FilKup08a} (see also \cite{HeaKu04}), the concept
of Pareto-optimality has been used to determine an equilibrium pricing
rule, where the term ``pricing rule'' refers to a finitely additive
measure (an element of the dual of $\linf$).  More precisely, the
authors provide sufficient and necessary conditions for the existence
of a Pareto optimal allocation of agents' endowments, from which an
equilibrium pricing rule is induced (in fact, the equilibrium pricing
rule is the super-gradient of the representative agent's risk
measure).

In this work, instead of establishing an equilibrium pricing rule from
a Pareto optimal allocation, we take a more direct approach and apply
the classical market-clearing arguments to derive a \textit{partial
  equilibrium price} for a bundle $\bB$ of claims (in addition to a
liquid incomplete financial market).  Provided that the agents are not
already in a Pareto-optimal configuration, an agent's demand of the
vector $\bB$ at a price $\bsp$ is defined as the number of units of
$\bB$ that the agent is willing to buy at price $\bsp$. An equilibrium
price for $\bB$ is, then, the price at which the sum of agents'
demands is equal to zero for each component of $\bB$ and the resulting
re-allocation of the bundle $\bB$ is called the \textit{partial
  equilibrium allocation}.  In section \ref{sec:PEPA}, we give
sufficient conditions for existence and uniqueness of the
partial-equilibrium price and allocation and we relate it to the
notion of {\em agents' agreement}. This result generalizes Theorem 5.8
in \cite{AntZit08}, where the case of two agents with exponential
utility functions is considered.

Having settled the problem of existence and uniqueness of the
partial-equilibrium price, we turn to the following question:
\emph{How is the equilibrium price-allocation affected by (small)
  perturbations of the agents' decision criteria?}  This problem is of
considerable importance, since estimation of the shape of each agent's
acceptance set is an extremely difficult task. It is therefore
reasonable - in the spirit of Hadamard's requirements (see
\cite{Had02}) - to ask that any result, which uses acceptance sets as
exogenously given, should satisfy adequate stability criteria. Despite
its importance, the problem of stability of equilibrium prices with
respect various problem primitives has not been previously studied in
the context of continuous-time finance.  Well-posedness of various
``single-agent'' optimization problems, on the other hand, has been
extensively studied and has always been an important part of the
optimization theory (standard references on stability, and, more
generally, well-posedness of variational problems are \cite{Luc06},
\cite{RocWet98} and \cite{DonZol93}).  However, stability of the
agent's investment decisions in the presence of a financial market has
only recently been investigated, and only for cases of utility
function maximizers (see \cite{CarRas05}, \cite{JouNap04},
\cite{KarZit07}, \cite{LarZit07}). Stability of problems related to
the more general notion of a convex risk measure has still not been
studied.  In our setting, as demonstrated in Theorem \ref{thm:PEPA},
the problem of existence of the partial-equilibrium price can be
viewed as a minimization problem of the sum of the agents' capital
requirements. Considered as such, its stability can be guaranteed by
certain conditions on allowed perturbations of the agents' acceptance
sets.  The central notion in this analysis is the one of
\textit{Kuratowski convergence}; it is applied to sets of
``acceptable'' numbers of units of the given bundle of claims and
provides a framework for sufficient conditions for stability.  As
special cases, we consider the set-ups of \cite{HugKra04} and
\cite{KarZit07}, where agents' risk preferences are modelled by
utility functions.

The structure of the paper is as follows: In section \ref{sec:setup},
we describe the market model, introduce necessary notation and state
some properties of the agents' acceptance sets and capital
requirements. In section \ref{sec:agreement}, we define and discuss
the notion of mutually agreeable claim-allocations and analyze its
relation to the Pareto optimality. Partial-equilibrium
price-allocation is introduced in section \ref{sec:PEPA}, where an
existence and uniqueness result is provided and discussed. Finally, in
section \ref{sec:stability} we exhibit conditions on specification of
the agents' acceptance sets that yield stability of the equilibrium
price.

\bigskip

\section{The Market Set-up}\label{sec:setup}

\subsection{The Liquid Part of the Financial Market}

Our model of the liquid part of the financial market is based on a
filtered probability space $(\Omega ,\mathcal{F},\mathbb{F },\PP)$,
$\mathbb{F}=\left( \mathcal{F}_{t}\right) _{t\in [0,T]}$, $T>0$, which
satisfies the usual conditions of right-continuity and
completeness. There are $d+1$ traded assets ($d\in\mathbb{N}$), whose
discounted price processes are modelled by an
$\mathbb{R}^{d+1}-$valued locally bounded semimartingale \label{stock
  price} $(S^{(0)}_t;\bS_t)_{t\in [0,T]} =(S^{(0)}_{t} ;
S^{(1)}_{t},\dots,S^{(d)}_{t}) _{t\in[ 0,T ] }$. The first asset
$S^{(0)}_{t}$ plays the role of a num\'{e}raire security or a discount
factor.  Operationally, we simply set $S^{(0)}_t\equiv 1$, for all
$t\in [0,T]$, $\PP-$a.s. We also impose the assumption of no free
lunch with vanishing risk (see \cite{DelSch94}). Namely, we define
\[\MM_a=\{\QQ\ll\PP:\bS \text{ is a local
  martingale under }\QQ\}\text{ and } \MM_e=\{\QQ\approx\PP:\bS \text{
  is a local martingale under }\QQ\}\] and make the following standing
assumption.
\begin{assumption}\label{ass:NA}
$\MM_e\neq\emptyset$.
\end{assumption}
\noindent We allow the possibility that the
liquid part of the financial market is incomplete, i.e., that
$\MM_e$ is not a singleton.

\subsection{Admissible Strategies}
For $\sigma$-algebra $\GG\subseteq \FF$, $\lzer(\GG)$ denotes the set of
all $\PP-$a.s.~equivalence classes of $\GG$-measurable random
variables, and $\linf(\GG)$ the set of all (classes of) essentially
bounded elements of $\lzer(\GG)$. When the underlying $\sigma$-algebra
$\GG$ is omitted, it should be assumed that $\GG=\FF$.  Shortcuts
$\BB+\CC=\sets{X+Y}{X\in\BB,\ Y\in\CC}$, $-\BB=\sets{-X}{X\in\BB}$,
$\BB_+=\{X\in\BB\,:\,X\geq 0, \text{ a.s.}\}$, $\BB_-=\{X\in\BB\,:\,
X\leq 0,\text{ a.s.}\}$ for $\BB,\CC\subseteq \lzer=\lzer(\FF)$, will
be used throughout.

A financial agent (with initial wealth $x$) invests in the market by
choosing a portfolio strategy $\bvt\in L(\bS)$, where $L(\bS)$ denotes
the set of predictable stochastic processes integrable with respect to
$\bS$.  The resulting \define{wealth process}, $\prf{X^{x,\bt}_t}$, is
simply the stochastic integral:
\begin{equation}
X^{x,\bt}_t=x+(\bvt\cdot \bS)_t=x+\int_0^t \bvt_u\, d\bS_u.
\end{equation}
We say that a strategy $\bvt$ \textit{admissible} if the induced
wealth process is uniformly bounded from below by a constant and we
denote the set of admissible strategies by $\bT$, i.e.,
\begin{equation}\label{admissible strategies}
  \bT=\{ \bvt\in L(\bS):\exists c\in\R \text{ such that }c
\leq (\bvt\cdot \bS)_t,\text{ } \forall t\in [0,T],\text{ a.s.} \}
\end{equation}
The collection of all wealth processes corresponding to the initial
wealth $x$ and admissible portfolio strategies is denoted by $\XX(x)$,
i.e.,
\[ \XX(x)=\sets{\prf{X^{x,\bt}_t}}{\bt\in \bT}.\]
Furthermore,
we define the sets $\mathcal{X}=\underset{x\in\R}{\bigcup}\X{x}$,
$\xinf=\mathcal{X}\cap\linf$
and $\ar=\sets{X\in\mathcal{X}}{-X\in\mathcal{X}}$.
\begin{remark}\
\begin{enumerate}
\item
Local
boundedness of the price process $\bS$ implies that $X\in\ar$ if and
only if there exists $x\in\R$ and $\bvt\in\bT$ such that
$X=x+(\bvt\cdot \bS)_T$ and $(\bvt\cdot \bS)_t$ is uniformly
bounded. In particular, $\ar\subseteq\linf$.
\item The lower bound on the losses of the admissible strategies is
  imposed to avoid pathologies that the so-called \textit{doubling
    strategies} create. Moreover, Assumption \ref{ass:NA} excludes the
  existence of arbitrage opportunities in the liquid market (see
  \cite{DelSch94}, Corollary 1.2). Note also that $X\in\xinf$ does not
  imply that $-X\in\mathcal{X}$, since there exist admissible
  strategies such that $(\bvt\cdot \bS)_T\in\linf$ but $(-\bvt\cdot
  \bS)_t$ is not uniformly bounded from below.
\end{enumerate}
\end{remark}

\subsection{The Acceptance Sets}
Given the financial market $(S^{(0)};\bS)$ and the set of admissible
strategies $\bT$, we suppose that each agent's risk preferences,
investment goals, possible stochastic income, etc., are
incorporated in a set $\tAA\subseteq\lzer(\FF)$ called the
\textit{acceptance set}.
We interpret $\tAA$ as the set that
contains the discounted net wealths of investment positions with
maturity  $T$ that the agent deems  acceptable at time $t=0$.

In concordance with the standard postulates of the risk-measure
theory, we assume that $\tAA$ satisfies the following axioms:
\begin{itemize}
    \item [\textit{Ax1.}]
%For $B,C\in\lzer(\FF)$, if $B\in\tAA$ and $B\leq
%    C$, then $C\in\tAA$.
      $\tAA+\lzer_+\subseteq\tAA$.
%, where
%      $\tAA+\lzer_+=\sets{B+X}{B\in\tAA,\, X\in\lzer_+}$,
% If also $B,C\in\linf$, then $B\leq
%     C$,
%     $\PP-$a.s.~and $B\in\tAA$ implies that
%     $C\in\tAA$.
    \item [\textit{Ax2.}] $\tAA$ is convex.
    \item [\textit{Ax3.}] $\tAA\cap\lzer_{-}(\FF)=\{0\}$.
    \item [\textit{Ax4.}] $\tAA-\XX(0) \subseteq \tAA$. \label{axiom 4}.
\end{itemize}
For future use we set $\AA=\tAA\cap\linf$.
\begin{remark}
\label{rem:on-axioms}
Axiom \textit{Ax1} simply states that every investment with payoff
a.s.~above the payoff of an acceptable claim is also
acceptable.  Axiom \textit{Ax2} reflects the fact that diversified
portfolios of acceptable investments should also be acceptable, while
Axiom \textit{Ax3} means that the ``status quo'' (i.e., no investment
at all) is an acceptable position and that the non-trivial investments
which never make money are not acceptable.  Finally, axiom
\textit{Ax4} is the one that provides a link between the liquid market
and the agent's acceptable positions. One should think of
$\tAA$ as an ``already-optimized'' representation of agent's
preferences, in the sense that the fact that the liquid market stands
at the agent's disposal has already been taken into account.
One of the direct consequences of {\em Ax4}, and a more mathematical
reformulation of the last sentence, is the following property:
\begin{equation}\label{equ:property4'}
\text{If there exists }X\in\X{0} \text{ such that }
B+X\in\tAA, \text{ then } B\in\tAA.
\end{equation}
More directly, if a position can be improved to acceptability by
costless trading, it should already be considered acceptable.  The
reader should note that the situation is not entirely symmetric: it
can happen that $B-X$ is acceptable for some $X\in\XX(0)$, but $B$ is
not. The reason is that $X$ may not be bounded from above so that
there is no admissible strategy (with $-X\not\in\XX(0)$ being the
prime candidate) which will bring $B$ into acceptability.
\end{remark}
An important, but by no means only, example of an acceptable set which
satisfies {\em Ax1-Ax4} can be constructed using utility functions:
\begin{example}\label{exp:utility}
  A classical example of an acceptance set that satisfies the axioms
  \textit{Ax1}-\textit{Ax4} is the one induced by a utility function,
  i.e., a mapping $U:(a,\infty) \to\R$, $a\in [-\infty,0]$, which is
  strictly concave, strictly increasing, continuously differentiable
  and satisfies the Inada conditions
$$\underset{x\to a^+}{\lim}U'(x)=+\infty
\text{  and  }\underset{x\to +\infty}{\lim}U'(x)=0.$$
We also include a random endowment (illiquid investments, stochastic
income) whose value at time $T$ given by $\EN\in\linf(\FF_T)$. The
agent's investment goal is to maximize the expected utility by trading
the market assets and for every contingent claim
$B\in\lzer_+-\linf_+$, the resulting indirect utility is defined by
\begin{equation}\label{equ:indirect utility}
    u(x|B)=\underset{X\in\mathcal{X}(0)}{\sup}\EE[U(x+\EN+X+B)],
\end{equation}
where $x>0$ is the agent's initial wealth. For sufficient assumptions
that lead to the existence of the optimal trading strategy, we refer
the interested reader to \cite{CviSchWan01}, \cite{HugKra04}
for the case $a>-\infty$ and \cite{BiaFri08}  and \cite{OweZit06}
for $a=-\infty$.
  The
set of acceptable claims is then given by
\begin{equation}\label{equ:utility acc set}
    \tAA_U(x)=\sets{B\in\lzer_+-\linf_+}{u(x|B)\geq u(x|0)}.
\end{equation}
It is straightforward to check that $\tAA_U(x)$ indeed satisfies the
axioms \textit{Ax1}-\textit{Ax4}, for $x>0$.
\end{example}

\subsection{The Convex Capital Requirement}
Given an acceptance set $\tAA$, we call the map
$\ra:\linf\to\bar{\R}$, defined by
\begin{equation}\label{risk measure definition}
\ra(B)=\inf\{m\in\R:m+B\in\tAA \}, \text{  for every
}B\in\linf,
\end{equation}
the agent's \textit{convex capital requirement} or
\textit{convex risk measure} induced by the acceptance set $\tAA$.
It follows that $\ra(\cdot)$ is
convex, non-increasing and cash invariant, i.e., $\ra(B+m)=\ra(B)-m$,
for every $B\in\linf$ and $m\in\R$.

Axioms \textit{Ax1} and \textit{Ax2} imply that $\ra(0)=0$ and the
inequality $-\|B\|_{\infty}\leq B\leq\|B\|_{\infty}$ together with
axiom \textit{Ax1} force $\rho_{\AA}(B)\in
[-\|B\|_{\infty},\|B\|_{\infty}]\subseteq\R$ for every $B\in\linf$.
The inclusion $\AA\subseteq\sets{B\in\linf}{\ra(B)\leq 0}$ holds
trivially. If, in addition, the set $\AA$ satisfies the following mild
closedness property
\begin{equation}\label{equ:propertyax5}
  \sets{\lambda\in[0,1]}{\lambda
    m+(1-\lambda)B\in\AA}
  \text{ is closed in }[0,1],
  \text{ for every $m\in\R_+$  and $B\in\linf$,}
\end{equation}
the inverse inclusion also holds (see Proposition 4.7 in
\cite{FolSch04}).
Property (\ref{equ:propertyax5}) holds, in particular, if
$\tAA\cap V$ is closed (with respect to any linear topology) for any finite-dimensional subspace
$V\subseteq\linf$.
% Summing up, if the acceptance set $\tAA$ satisfies the
% axioms \textit{Ax1}-\textit{Ax4} and property
% (\ref{equ:propertyax5}), the map $\ra(\cdot)$ defined in
% (\ref{risk measure definition}) is a convex risk measure, for
% which $\ra(0)=0$, and the intersection $\AA=\tAA\cap\linf$
% can be recovered from $\ra(\cdot)$ through the equality
% $\AA=\sets{B\in\linf}{\ra(B)\leq 0}$.
In what follows, with a slight abuse of terminology, when we
mention the term acceptance set we will refer to the set
$\AA=\tAA\cap\linf$, for $\tAA$ that satisfies
\textit{Ax1}-\textit{Ax4}.
\begin{remark}\label{rem:other rms}
Similar definitions of the convex capital requirement have been
given in \cite{FolSch04} (page 207) and \cite{FriSca06}. In the
former, a given acceptance set $\AA$ is related to the market
through a larger acceptance set $\hat{\AA}$, defined by
\begin{equation}
\hat{\AA}=\{B\in\linf:\exists\, \bvt\in\bT, A\in\AA \text{ such that
}(\bvt\cdot\bS)_T+B\geq A, \PP-\text{a.s.}\}.
\end{equation}
In our case, \eqref{equ:property4'} implies that $\hat{\AA}=\AA$,
which is yet another reformulation of the ``already-optimized'' property of
Remark \ref{rem:on-axioms}.
In \cite{FriSca06}, the authors define the generalized capital
requirement by
\begin{equation*}
\hat{\rho}_{\AA}(B)=\inf\sets{m\in\R}{\exists\, X\in\X{m}\text{ such
that } X+B\in\AA}.
\end{equation*}
If the acceptance set $\tAA$ satisfies the axioms
\textit{Ax1}-\textit{Ax4}, it is straightforward to show that $\ra(B)=
\hat{\rho}_{\AA}(B)$.  The existence of an admissible strategy in the
definitions of $\hat{\rho}_{\AA}(\cdot)$ and $\hat{\AA}$ has been
established in \cite{Xu06}, Theorem 2.6.
\end{remark}

\subsection{A Robust Representation}
It is shown in \cite{FolSch02} that under the assumption that $\AA$ is
weak-$\ast$ closed (closed in the weak topology
$\sigma(\linf,\lone)$), the convex risk measure $\ra(\cdot)$ admits a
robust representation in the sense of \cite{ArtDelEbeHea99} and
\cite{Del02}. The additionally imposed axiom \textit{Ax4} provides
some further information about the penalty function and, in
particular, about its effective domain, denoted by $\Ma$. The
following proposition is similar, but not identical, to the results in
\cite{FolSch04} and \cite{KloSch07}.
\begin{proposition} If $\AA$ is a weak-$\ast$ closed acceptance set, then
\begin{enumerate}
\item  $\ra$ admits a robust representation of the following form
    \begin{equation}\label{equ:representation}
\ra(B)=\underset{\QQ\in\MM_a}{\sup}\{\EE^{\QQ}[-B]-\alpha_{\AA}(\QQ)\}
    \end{equation}
for every $B\in\linf$, where
$\alpha_{\AA}(\QQ)=\underset{B\in\AA}{\sup}\{\EE^{\QQ}[-B]\}$,
i.e., $\Ma\subseteq\MM_a$, and
\item the set of measures, denoted by $\partial\rho_{\AA}(B)$, at
  which the supremum in (\ref{equ:representation}) is attained, is
  non-empty.
\end{enumerate}
\end{proposition}
\begin{proof}
  Thanks to the results in \cite{FolSch04} and \cite{KloSch07}, it is
  enough to show that for every $\QQ\notin\MM_a$,
  $\alpha_{\AA}(\QQ)=+\infty$. For every such $\QQ$, there exists an
  admissible terminal wealth $X\in\X{x}$, such that $\EE^{\QQ}[X]>x$,
  i.e., there exists a portfolio $\bvt\in\bT$, such that $(\bvt\cdot
  \bS)_t $ is uniformly bounded from below and $\EE^{\QQ}[(\bvt\cdot
  \bS)_T]>0$ (see Theorem 5.6 in \cite{DelSch94}). Then, for every
  $k\in\mathbb{N}$, we define $B_k=-((\bvt\cdot \bS)_T\wedge k)$,
  which belongs to $\linf$.  Hence, $B_k+(\bvt\cdot \bS)_T=((\bvt\cdot
  \bS)_T-k)\mathbf{1}_{\{(\bvt\cdot \bS)_T\geq k\}}\geq 0$, which
  means that $B_k+(\bvt\cdot \bS)_T\in\tAA$ for every
  $k\in\mathbb{N}$. Also, by \eqref{equ:property4'} we have that
  $\lambda B_k\in\AA$, for all $\lambda >0$. Thus,
  $\alpha_{\AA}(\QQ)\geq \EE^{\QQ}[-\lambda B_k]$, for every
  $k\in\mathbb{N}$. Finally, it is enough to first let $k\to\infty$
  and use the Monotone convergence theorem to get that
  $\alpha_{\AA}(\QQ)\geq \lambda\EE^{\QQ}[(\bvt\cdot \bS)_T ]$, and
  then let $\ld\to\infty$.
\end{proof}
%Note that since $\ra(0)=0$,
%$\underset{\QQ\in\Ma}{\min}\{\alpha_{\AA}(\QQ)\}=0$.
\begin{corollary}\label{cor:repli-invariance}
  If $\AA$ is a weak-$\ast$ closed acceptance set, then $\ra(\cdot)$
  satisfies the following \textit{replication invariance} property:
  for every $B\in\linf$ and every $C\in\ar\cap \XX(x)$, we have
  $\ra(B+C)=\ra(B)-x.$
\end{corollary}

\subsection{Risk-equivalence}
The following definition (see also
\cite{AntZit08}) will be used extensively in the sequel:
\begin{definition}
  Two random variables $B,C\in\linf$ are said to be {\em
    risk-equivalent} (or \textit{equivalent with respect to risk}),
  denoted by $B\sim C$, if $B-C\in\ar$.
\end{definition}

It is straightforward to check that the relation $\sim$ is indeed an
equivalence relation in $\linf$.  The condition $B\sim C$ means that
the claims with payoffs $B$ and $C$ carry the same unhedgeable
risk. Moreover, it is easy to see that the condition $B\sim C$ implies
that
\begin{equation}
   \label{equ:linear}
   % \nonumber
   \begin{split}
     \forall\, \lambda\in[0,1],\ \ra(\lambda
     B+(1-\lambda)C)=\lambda\ra(B)+(1-\lambda)\ra(C).
   \end{split}
\end{equation}

On the other hand, if $B\nsim C$, convex combinations of the payoffs
$B$ and $C$ may lead to reduction of risk. If any such combination of
claims (which do not belong in the same equivalence class) reduces the
risk, the corresponding acceptance set $\AA$ is called
\textit{risk-strictly convex}:
\begin{definition}\label{def:strict-convex}
  An  acceptance set $\AA$ is called \textit{risk-strictly
    convex} if
for all $B,C\in\AA$ with $B\nsim C$ and every $\ld\in (0,1)$,
  there exists a random variable $E\in\linf_+$ and
  $\QQ\in\partial\rho_{\AA}(\ld B + (1-\ld) C)$ such that, $\QQ(E>0)>0$
  and $$\lambda B+(1-\lambda)C-E\in\AA.$$
\end{definition}

\begin{proposition}\label{pro:strictly convex r.m.}
  Let $\AA$ be a weak-$\ast$ closed acceptance set. Then, $\AA$ is
  risk-strictly convex if and only if for all $B,C\in\linf$ the
  condition (\ref{equ:linear}) implies $B\sim C$.
\end{proposition}
\begin{proof}
  We first assume that $\AA$ is risk-strictly convex and show the
  contrapositive of the stated implication. For arbitrarily chosen
  $B,C\in\linf$ such that $B\nsim C$, we have that
  $B+\ra(B),C+\ra(C)\in\AA$. Hence, for a given $\lambda\in(0,1)$,
  there exists $E\in\linf_+$ and $\QQ\in\partial\rho_{\AA}(\ld B +
  (1-\ld) C)$ such that $\QQ(E>0)>0$ and $$\lambda B+(1-\lambda)C+
  \lambda\ra(B)+(1-\lambda)\ra(C)-E\in\AA.$$ This implies that
  $\ra(\lambda B+(1-\lambda)C-E)\leq \lambda\ra(B)+(1-\lambda)\ra(C)$,
  and so, by monotonicity of $\rho_{\AA}$, we have
\begin{eqnarray*}
  \ra(\lambda B+(1-\lambda)C) &=&
  \EE^{\QQ}[-\lambda
  B-(1-\lambda)C]-\alpha_{\AA}(\QQ)<  \EE^{\QQ}[-\lambda
  B-(1-\lambda)C+E]-\alpha_{\AA}(\QQ)\\
  &\leq & \underset{\tilde{\QQ}\in\MM_{\AA}}{\sup}\{\EE^{\tilde{\QQ}}[-\lambda
  B-(1-\lambda)C+E]-\alpha_{\AA}(\tilde{\QQ})\}
  = \ra(\lambda
  B+(1-\lambda)C-E)\\
  &\leq & \lambda\ra(B)+(1-\lambda)\ra(C).
\end{eqnarray*}
Conversely, suppose that (\ref{equ:linear}) implies $B\sim C$, for all
$B,C\in\linf$.  Then for any pair $B\nsim C$ and every
$\lambda\in(0,1)$ we must have that $$\ra(\lambda
B+(1-\lambda)C)<\lambda\ra(B)+(1-\lambda)\ra(C),$$ so it is enough to
take $E=-\ra(\lambda B+(1-\lambda)C)$ in the definition of risk-strict
convexity.
\end{proof}
\begin{remark}
An examination of the above proof reveals that the seemingly stronger
condition where the random variable $E$ is replaced by a positive
constant leads to the same concept as in Definition
\ref{def:strict-convex}.
\end{remark}
Under the assumption that the acceptance set $\AA$ is
risk-strictly convex, we can say a bit more about the effective
domain of the penalty function of the induced risk measure,
$\MM_{\AA}$.

\begin{proposition}\label{pro:Q(0) is equivalent.}
If the acceptance set $\AA$ is weak-$\ast$ closed and risk-strictly convex,
then $\MM_{\AA}\subseteq\MM_e$.
\end{proposition}
\begin{proof}
  Suppose, to the contrary, that there exists
  $\QQ\in\MM_{\AA}\setminus \MM_e$. Since $\MM_{\AA}\subseteq\MM_a$,
  there exists $A\in\FF$ be such that $\QQ[A]=0$, but
  $\PP[A]>0$. Then, since $ \alpha_{\AA}(\QQ)=\sup_{C\in\linf} \left(
    \EE^{\QQ}[-C]-\ra(C)\right)$, we have \[\ra(C-n\ind{A})\leq
  \EE^{\QQ}[C]+\alpha_{\AA}(\QQ)<\infty,\] for all $n\in\N$, and all
  $C\in\linf$. By convexity,
\[
\rho(C-\ind{A}) \leq \tfrac{1}{n}\rho(C-n\ind{A})+ (1-\tfrac{1}{n})
\rho(C) \leq \tfrac{1}{n}( \EE^{\QQ}[C]+\alpha_{\AA}(\QQ) )
+(1-\tfrac{1}{n}) \rho(C)),
\]
for each $n\in\N$, so $\rho(C-\ind{A})=\rho(C)$. It follows now from
Proposition \ref{pro:strictly convex r.m.} that $\ind{A}\in\ar$, which
is in contradiction with the assumption of No Free Lunch with
Vanishing Risk.
\end{proof}

\begin{remark}
  In the terminology of \cite{FolSch04} (see page 173), a risk measure
  $\rho$ is called \emph{sensitive} if $\rho(-B)>0$ for every
  $B\in\linf_+\backslash\{0\}$.  The fact, as stated in Proposition
  \ref{pro:Q(0) is equivalent.}, that the minimizers of the penalty
  function $\alpha_{\AA}(\cdot)$ are equivalent to $\PP$ when $\AA$ is
  risk-strictly convex, implies that a risk-strictly convex risk
  measures are sensitive. Indeed, by (\ref{equ:representation}),
  $\ra(-B)= \EE^{\QQ}[B]-\alpha_{\AA}(\QQ)$, for any $\QQ\in\partial
  \rho_{\AA}(0)$.  Clearly $\EE^{\QQ}[B]>0$ and
  $\alpha_{\AA}(\QQ)(B)\leq 0$, so $\ra(-B)>0$.
\end{remark}
Another direct property of risk-strictly convex acceptance
set is the following:
\begin{proposition}\label{pro:sum-of-rm-positive}
  Let $\rho_{\AA}$ be the risk measure corresponding to a
  weak-$\ast$ closed, risk-strictly convex acceptance set
  $\AA$. Then, $$\ra(B)+\ra(-B)>0,\text{ for }
  B\in\linf\setminus\ar.$$ In particular, $\AA\cap
  (-\AA)=\mathcal{X}(0)\cap\ar$.
\end{proposition}
\begin{proof}
For any such $B\in\linf\setminus\ar$, $\rho_{\AA}(B)+B\in\AA$ and
$\rho_{\AA}(-B)-B\in\AA$. By assumption (since $2B\notin\ar$),
there exists $E\in\linf_+\backslash\{0\}$ such that
$$\frac{1}{2}(\rho_{\AA}(B)+\rho_{\AA}(-B))-E\in\AA.$$ Hence,
$\frac{1}{2}(\rho_{\AA}(B)+\rho_{\AA}(-B))-\ra(-E)\geq 0$ and by
monotonicity of $\ra$, we get that $\rho_{\AA}(B)+\rho_{\AA}(-B)>0$.
The last statement follows from Corollary \ref{cor:repli-invariance}.
\end{proof}

\bigskip

\section{Mutually-Agreeable Bundles}\label{sec:agreement}

\subsection{The Agents}
We consider $I\geq 2$ financial agents and suppose that each agent $i$
has access to a sub-market $\bS_i$ of $\bS$, i.e., she is allowed to
invest only in $(S^{(0)}_{t} ;
S^{(j^i_1)}_{t},\dots,S^{(j^i_{d_i})}_{t}) _{t\in[ 0,T ] }$, where
$1\leq j^i_1 < \dots < j^i_{d_i}\leq d$. Note that the num\'{e}raire
$S^{(0)}\equiv 1$ is accessible to each agent. In order to take the
whole market into account and avoid trivialities, we also assume that
each component of $\bS$ is accessible to at least one agent.  Note
that the Assumption \ref{ass:NA} implies that $\MM_a^i\not=\emptyset$,
for all $i$, where
\begin{equation*}
\MM_a^i=\sets{\QQ\ll\PP}{\bS_i \text{ is a local
martingale under }\QQ},\  i=1,2,\dots, I.
\end{equation*}
We define the sets $\XX_i$, $\XX_i(x)$, $\bT_i$ and $\ar_i$ exactly as
in section \ref{sec:setup}, with $\bS_i$ used in lieu of
$\bS$. Moreover, each agent is assumed to have an acceptance set
$\tAA_i$ which satisfies the axioms \textit{Ax1}-\textit{Ax4}. The
induced risk measure $\rho_{\AA_i}$ on $\linf$ will be denoted by
$\rho_i$, and $\MM_i$ will be the shortcut for $\MM_{\AA_i}$, i.e., it
will stand for the effective domain of the corresponding penalty
function $\alpha_i$, $i=1,2,\dots, I$.  If we further assume that the
intersection $\tAA_i\cap\linf$, denoted by $\AA_i$, is weak-$\ast$
closed and hence the induced risk measure $\rho_i=\rho_{\AA_i}$ admits the
following robust representation
\begin{equation}\label{equ:irob-repr}
\rho_{i}(B)=\underset{\QQ\in\MM_i}{\sup}\{\EE^{\QQ}[-B]-\alpha_{i}(\QQ)\},
\end{equation}
where $\MM_i\subseteq\MM^{i}_a$ for all $i$.
As above,  $\partial\rho_{i}(B)$ denotes the set of all  maximizers
in \eqref{equ:irob-repr},
for $B\in\linf$ and $i=1,2,...,I$.

\subsection{Bundles, Allocations and Agreement}
% The acceptance sets $\AA_i$'s depend, among others, on the current
% investments of the agents. Hence, every additional risk undertaken
% by an agent should change her acceptance set. This implies that the
% price (induced by the risk measure) of every claim $B\in\linf$
% depends on the undertaken position on other claims too. Thus, it is
% more realistic for every pricing rule to consider a vector of claims
% instead of a single one. For this reason, we suggest a pricing
% scheme for vectors of $n$ contingent claims $\bB\in(\linf)^{n}$, for
% some $n\in\mathbb{N}$.
For bundle of claims $\bB=(B_1,B_2,...,B_n)\in(\linf)^{n}$,
$n\in\mathbb{N}$, a matrix $(a_{i,k})=\boa\in\R^{I\times n}$ is called
a \emph{feasible allocation} or simply an \textit{allocation}, if
$\sum_{i=1}^I a_{i,k} = 0$ for all $k=1,2,...,n$. For convenience, the
$i$-th row $(a_{i,k})_{k=1}^n$ of $\boa$ will be denoted by $\boa_i$ ;
it counts the quantities of each of the $n$ components of $\bB$ held
by the agent $i$.  The set of all feasible allocations is denoted by
$\mathbf{F}$, i.e.,
\begin{equation}\label{equ:def_of_allo}
\mathbf{F}=\{\boa\in\R^{I\times n}\text{ such that }\sum_{i=1}^I
\boa_i = (0,0,...,0)\}.
\end{equation}
We usually think of the elements of $\bB$ as the claims (typically not
replicable in the liquid market $\bS$) the agents are trading among
themselves.  These claims are in zero net supply, i.e., some of the
agents will be taking positive and some negative positions in them.
Clearly, the agents will be willing to share the bundle of claims
$\bB$ according to an allocation $\boa\in\mathbf{F}$ only if there
exists a price vector $\bsp\in\R^n$, for which the position
$\boa_i\cdot\bB-\boa_i\cdot\bsp$ is acceptable for each agent
$i$. More precisely, we give the following definition:
\begin{definition}\label{def:gagreement}
  The pair $(\bB,\boa)\in(\linf)^{n}\times\mathbf{F}$ of a bundle of
  claims and an allocation is called \emph{mutually
    agreeable}\index{mutually agreeable} if there exists a (price)
  vector $\bsp\in\R^n$ such that
  $\boa_i\cdot\bB-\boa_i\cdot\bsp\in\AA_i$, for all $i=1,2,...,I$.
\end{definition}
For an allocation $\boa$, let $\agset^{\boa}$ denote the set of all
feasible allocations of $\bB$, acceptable for every agent:
$$\agset^{\boa}=\{\bB\in (\linf)^n:\text{ } (\bB,\boa)\text{ is
  mutually agreeable}\}.$$ We also set
$\hr_{\boa}=\{\bB\in(\linf)^n:\boa_i\cdot \bB\in\ar_i,\text{ } \forall
i=1,2,...,I\}$, for $\boa\in\mathbf{F}$.
\begin{proposition}\label{pro:first-prop-gG}
  For every allocation $\boa\in \mathbf{F}$, $\agset^{\boa}$ is
  convex. If, additionally, $\AA_i$ is weak*-closed and risk-strictly
  convex for every $i$, then $$\agset^{\boa}\cap
  (-\agset^{\boa})=\hr_{\boa},\text{ for all $\boa\in\mathbf{F}$}.$$
\end{proposition}
\begin{proof}
The convexity follows directly from the convexity of $\AA_i$'s.

For the second statement, suppose first
that $\bB\in\agset^{\boa}\cap
(-\agset^{\boa})$, i.e., there exist $\bsp, \hat{\bsp}\in\R^n$ such
that $\boa_i\cdot\bB-\boa_i\cdot\bsp\in\AA_i$ and
$-\boa_i\cdot\bB+\boa_i\cdot\hat{\bsp}\in\AA_i$, for all $i$. The
convexity of $\AA_i$ then implies that
$\frac{1}{2}\boa_i\cdot(\hat{\bsp}-\bsp)\in\AA_i$, which yields that
$\boa_i\cdot(\hat{\bsp}-\bsp)\geq 0$, for all $i=1,2,...,I$. Since
$\sum_{i=1}^I\boa_i=0$, we conclude that for every agent
$\boa_i\cdot\bsp=\boa_i\cdot\hat{\bsp}$.  It follows that
$\rho_i(\boa_i\cdot\bB)\leq-\boa_i\cdot\bsp$ and also
$\rho_i(-\boa_i\cdot\bB)\leq\boa_i\cdot\bsp$ for all $i$. Consequently,
$\rho_i(\boa_i\cdot\bB)+\rho_i(-\boa_i\cdot\bB)\leq 0$,
which means (thanks to the risk-strict convexity of $\rho_i$) that
$\boa_i\cdot\bB\in\ar_i$ for all $i$, i.e., $\bB\in\hr_{\boa}$.

Conversely, let $\bB\in \hr_{\boa}$ so that $\boa_i\cdot \bB\in\ar_i$,
for all $i$.  We pick an arbitrary $\QQ\in\MM_a$ and define
$\bsp=\EE^{\QQ}[\bB]=(\EE^{\QQ}[\bB_1],\dots,
\EE^{\QQ}[\bB_n])\in\R^n$. Since $\QQ\in\MM^a_i$ for all $i$, we have
$\boa_i\cdot \bB\in \XX_i( \boa_i\cdot \bsp)$.  By Corollary
\ref{cor:repli-invariance}, $\boa_i\cdot
\bB-\boa_i\cdot\bsp+\in\AA_i$, for all $i$, so that
$\hr_{\boa}\subseteq \agset^{\boa}$. It remains to observe that
$\hr_{-\boa}=\hr_{\boa}$ and that $\agset^{-\boa}=-\agset^{\boa}$.
\end{proof}
For a bundle $\bB\in(\linf)^n$, the set $\agset^{\bB}$ is, in a sense,
polar to $\agset^{\boa}$:
$$\agset^{\bB}=\sets{\boa\in\mathbf{F}}{(\bB,\boa)
  \text{ is mutually agreeable}}.$$ Following the proof of Proposition
\ref{pro:first-prop-gG}, one can show that $\agset^{\bB}$ is convex,
closed in $\R^{n\times I}$ and that
$\boa\in\agset^{\bB}\cap\agset^{-\bB}$ implies that
$\bB\in\hr_{\boa}$.
% Theorem \ref{thm:PEPA} below states some sufficient conditions under
% which there exists an $\boa\in\mathbf{F}\setminus\{\mathbf{0}\}$,
% such that $\boa\in\agset^{\bB}$.
Another important notion in this setting is the inf-convolution of risk
measures, first introduced in \cite{BarElk04}.
\begin{definition}\label{def:inf-convo}
The \emph{inf-convolution} of the risk measures
$\rho_1,\rho_2,...,\rho_I$ is the map
$
\lzr:\linf\to\R\cup\{-\infty\}$, defined for $C\in\linf$ by
\begin{equation}\label{equ:inf-conv}
  (
  \lzr)(C)=
  \inf\sets{\sum_{i=1}^I \rho_i(B_i)}{
    B_1,\dots,B_I\in\linf,\, \sum_{i=1}^I B_i=C}.
\end{equation}
\end{definition}

Let $\MM$ denote the intersection $\bigcap_{i=1}^{I}\MM_i$. The
following assumption is equivalent to $(\lzr)(0)>-\infty$ (see
\cite{BarElk05}):
\begin{assumption}\label{ass:not empty intersectio}
$\MM\neq\emptyset$.
\end{assumption}
\begin{remark}
  Thanks to the assumption that each component of $\bS$ is available
  to at least one agent, we have that
  $\bigcap_{i=1}^{I}\MM_e^i=\MM_e$. Hence, if $\AA_i$'s are
  risk-strictly convex, it holds that $\MM\subseteq\MM_e$ and hence
  Assumption \ref{ass:not empty intersectio} is a strengthening of
  Assumption \ref{ass:NA}.
\end{remark}
The following Proposition is a mild generalization of Theorem
3.6 in \cite{BarElk05}, where the case of $I=2$ is addressed. The
proof when $I\geq 2$ is similar and hence omitted.
\begin{proposition}\label{pro:inf-convolution}
If $\AA_i$ is
weak-$\ast$ closed for every $i=1,2,...,I$, the Assumption \ref{ass:not empty intersectio}
implies that the map
$\rho_1\lozenge\dots \lozenge \rho_I:\linf\to\R$ is a convex risk measure, with penalty
function $h(\QQ)=\sum_{i=1}^I \alpha_i(\QQ)$ whose effective
domain is $\MM$.
\end{proposition}
\begin{definition}\label{def:Pareto}
We say that the agents are in a \textit{Pareto-optimal configuration} if
\[(\lzr)(0)=0.\]
\end{definition}

In words, Pareto optimality implies that there is no wealth-preserving
transaction that will be acceptable for everyone and strictly
acceptable for at least one agent.  The problem of Pareto-optimality
is closely related to the problem of \textit{optimal risk sharing}
(sometimes called \textit{Pareto optimal allocation}), which was
recently addressed by many authors in the cases where agents use
convex risk measures to value claim payoffs (see \cite{BarElk05}, \cite{BarSca08}, \cite{BurRus07}, \cite{FilKup08}, 
\cite{HeaKu04}, \cite{JouSchTou06}). Below, we state a well-known characterization of the
Pareto optimality in terms of the minimizers of the penalty functions
$\alpha_i$. We remind the reader that $\partial\rho_i(B)$ stands for
the {\em set} of maximizers in the robust representation
(\ref{equ:irob-repr}) of $\rho_i(B)$ and omit the standard proof:
\begin{proposition}\label{pro:Pareto and measures}
The agents are in a Pareto-optimal configuration if and only if
\[\cap_{i=1}^I \partial\rho_i(0)\not=\emptyset.\]
\end{proposition}
% \begin{proof}
% We observe that
% \begin{eqnarray*}
% % \nonumber to remove numbering (before each equation)
%   \rho_1\lozenge\dots \lozenge \rho_I(0)&=&\underset{\QQ\in\MM}{\sup}\{-h(\QQ)\}\\
%   &=& \underset{\QQ\in\MM}{\sup}\{-\sum_{i=1}^{I}\alpha_i(\QQ)\}=
% -\underset{\QQ\in\MM}{\inf}\{\sum_{i=1}^{I}\alpha_i(\QQ)\}\leq 0.
% \end{eqnarray*}
% and the equality holds if and only if $\alpha_i$'s have a common
% minimizer.
% \end{proof}

The following proposition states that if the agents are in
Pareto-optimal configuration, the risk-strict convexity assumption
implies that transactions involving non-replicable claims result
in strictly increased risk for at least one of the agents involved in this
transaction.
\begin{proposition}\label{pro:Pareto and Agreement}
Assume that $\AA_i$ are weak-$\ast$ closed and risk-strictly convex for
all $i=1,2,..,I$ and suppose that $(\lzr)(0)=0$. For any
choice of $B_1,B_2,...,B_{I}$ in $\linf$ with $\sum_{i=1}^I B_i=0$, it holds that
$$\sum_{i=1}^{I}\rho_i(B_i)=0\text{ if
and only if }B_i\in\ar_i, \text{  for all  } i=1,2,...,I.$$
\end{proposition}
\begin{proof}
  Assume that there exists $k\in\{1,2,...,I\}$ such that
  $B_k\notin\ar_k$. Then, by risk-strict convexity, for each $
  \lambda\in (0,1)$ there exists $E\in\linf_+\backslash\{0\}$ such
  that $\rho_k(\lambda B_k-E)\leq\lambda\rho_k(B_k)$. This implies
  that $\rho_k(\lambda B_k)<\lambda\rho_k(B_k)$. Since $\rho_i(\lambda
  B_i)-\lambda\rho_i(B_i)\leq 0$, $\forall i=1,2,...,I$, we have
$$\sum_{i=1}^{I}\rho_i(\lambda B_i)<\lambda
\sum_{i=1}^{I}\rho_i(B_i)=0,$$ which contradicts the assumption
$(\lzr)(0)=0$. The converse implication follows from the fact that
when $B_i\in\ar_i$, we have $\rho_i(B_i)=-\EE^{\QQ}[B_i]$, for any
$\QQ\in\MM^i_a$.
\end{proof}

\begin{corollary}
  Assume that all $\AA_i$ are weak-$\ast$ closed and risk-strictly
  convex and pick
  $\boa\in\mathbf{F}$ and $\bB\in\agset^{\boa}$.
If $\cap_{i=1}^I \partial\rho_i(0)\not=\emptyset$ then,
  $\boa_i\cdot \bB\in\ar_i$, for every $i=1,2,...,I$.
\end{corollary}

\begin{proof}
  $\bB\in\agset^{\boa}$ means that there exists a price vector $\bsp$,
  such that $\boa_i\cdot\bB-\boa_i\cdot\bsp\in\AA_i$, for all $i$.
  This implies that $\sum_{i=1}^{I}\rho_i(\boa_i\cdot\bB)\leq 0$,
  which, by the hypotheses and Proposition \ref{pro:Pareto and
    Agreement} yields that $\boa_i\cdot \bB\in\ar_i$ for all
  $i=1,2,...,I$.
\end{proof}

\begin{example}
  Suppose that all $I$ agents are exponential-utility maximizers with
  possibly different risk-aversion coefficients $\gamma_i$,
  $i=1,2,\dots,I$, i.e., $U_i(x)=-\exp(-\gamma_i x)$ (for details on
  the set of admissible strategies, we refer the reader to
  \cite{DelGraRheSamSchStr02} and \cite{ManSch05}). Let $\EN_i$,
  $i=1,2,\dots, I$, denote the agents' random endowments.  If we
  follow the arguments of Proposition 3.15 in \cite{AntZit08}, we can
  conclude that in the case $\bS_i=\bS$ for all $i\in\{1,2,...,I\}$,
  the agents will be in the Pareto-optimal configuration if and only
  if $\frac{\gamma_i}{\gamma_j}\EN_i\sim\EN_j$, for all
  $i,j=1,2,...,I$. A special case of this condition occurs when the
  agents' random endowments are replicable. However, this is not the
  case when agents have access to different markets.  To see that, let
  us consider the case $I=2$, with $\bS_1\neq\bS_2$ and
  $\EN_1=\EN_2=0$. It follows from Theorem 2.2 in
  \cite{DelGraRheSamSchStr02}, that
  $\partial\rho_1(0)=\partial\rho_2(0)$ (both are singletons in this
  case) if and only if
  $\gamma_1(\bvt^{(0)}_1\cdot\bS_1)_T=\gamma_2(\bvt^{(0)}_2\cdot\bS_2)_T$,
  where $\bvt^{(0)}_i$ is the optimal trading strategy in the market
  $\bS_i$ of the agent $i$, $i=1,2$.
\end{example}

\begin{example}
  The case where agents use power utility function,
  $U(x)=\frac{1}{p}x^p$, where $p\neq 0$ is the relative risk aversion
  coefficient, is similar to the previous example. Without going into
  details, let us mention that it is known (see for instance
  \cite{Sch04}) that if all agents have access to the same market,
  $\partial\rho_i(0)=\partial\rho_{j}(0)$, for all $i,j$ if and only if the agents'
  relative risk aversion coefficients are equal, regardless of their
  initial wealths. However, when agents have access to different
  markets, it is easy to construct counterexamples.
\end{example}

\bigskip

\section{The Partial Equilibrium Price Allocation}\label{sec:PEPA}

\subsection{The Demand Correspondence}
Having introduced the setup consisting of $I$ agents, their acceptance
sets and accessible assets, we turn to the partial-equilibrium pricing
problem for a fixed bundle of claims $\bB\in(\linf)^n$.  Our first
task is to analyze single agent's demand for $\bB$ under the natural
assumption that in a set of payoffs, an agent will choose the one
which minimizes the capital requirement. For linear combinations of
the components of $\bB$, we have the following, more precise,
definition:
% where $\bar{\R}=[-\infty,\infty]$ is the extended real line:
\begin{definition}\label{def:gdemand}
  For $i\in\{1,2,...,I\}$, the {\em agent $i$'s demand correspondence}
  $Z_i:\R^n\to 2^{\R^n}$ is defined by
\begin{equation}
  Z_i(\bsp)=\underset{\boa\in{\R}^n}{\argmin}\{\rho_i(\boa\cdot\bB-\boa\cdot
  \bsp)\}
  % =\underset{\boa\in{\R}^n}{\argmin}\{\rho_i(\boa\cdot\bB)+\boa\cdot
  % \bsp\}
\end{equation}
\end{definition}
\begin{definition}\label{def:strict-wrt-B}
  Let $\AA$ be a weak-$\ast$ closed acceptance set. For a bundle $\bB\in
  (\linf)^n$, we say that $\AA$ is \textit{strictly convex with
    respect to $\bB$} if for every
  $(\boa,m),(\boldsymbol{\delta},k)\in\AA(\bB)$, where
\begin{equation}\label{equ:AA(B)}
  \AA(\bB)=\{(\boa,m)\in\R^n\times\R:\boa\cdot\bB+m\in\AA\}
\end{equation}
such that $\boa\neq\boldsymbol{\delta}$ the following statement holds:

for every $\lambda\in (0,1)$ there exists a random variable
$E\in\linf_+$ and
$\QQ\in\partial\rho_i(\lambda\boa+(1-\lambda)\boldsymbol{\delta})\cdot\bB)$,
such that $\QQ (E>0)>0$ and
$$\lambda(\boa\cdot\bB+m)+(1-\lambda)(\boldsymbol{\delta}\cdot\bB+k)-E\in\AA.$$
\end{definition}
\begin{remark}
  Following the proof of Proposition \ref{pro:strictly convex r.m.}
  we can show that an acceptance set $\AA_i$ is strictly convex with
  respect to $\bB$ if and only if the function $\R^n\ni\boa\mapsto
  \rho_i(\boa\cdot\bB)$ is strictly convex.  Moreover, it is clear
  that the requirements of Definition (\ref{equ:AA(B)}) are implied by
  risk-strict convexity defined above whenever
  $\boa\cdot\bB\not\in\ar_i$ for $\boa\neq 0$.
\end{remark}
% \begin{definition}
%   \label{def:i-nonreplicable}
% For $i\in\set{\ft{1}{I}}$, a bundle $\bB\in(\linf)^n$ is called
% {\em $i$-nonreplicable} if $\boa\cdot\bB\in\ar_i$
% implies $\boa=0\in\R^n$.
% \end{definition}
Before we proceed to our next auxiliary result, we set
\begin{equation}
\nonumber
   \begin{split}
 \MM_i^{\bB} & =\bigcup_{\boldsymbol{\delta}\in\R^n}\partial
    \rho_i(\boldsymbol{\delta}\cdot \bB)\subseteq \MM_i,\
\mathcal{P}_i^{\bB}=\sets{\EE^{\QQ}[\bB]}{\QQ\in\MM_i^{\bB}}\subseteq\R^n.
   \end{split}
\end{equation}
\begin{lemma}\label{lem:sup_sets}
%Let $i$ be in $\set{1,2,\dots, I}$, and
%let $\bB\in(\linf)^n$ be $i$-nonreplicable.
Pick $i\in \set{1,2,\dots, I}$ and assume that
 that  $\AA_i$ is weak-$\ast$ closed and
strictly convex with respect to $\bB$. Let
$\sek{\boa}$ be a sequence of the form $\boa_k=\boa_0+ \gamma_k
\bsv\in\R^n$, for some $\bsv\in\R^n\setminus\set{0}$,
$\boa_0\in \R^n$, and a
sequence $\sek{\gamma}$ of positive constants with
$\lim_{k} \gamma_k=\infty$.
If  $\sek{\QQ}$ is  an
arbitrary
sequence of probability measures with
$\QQ_k\in\partial\rho_i(-\boa_k\cdot \bB)$, then
\begin{equation}\label{equ:general sup_sets}
\lim_{k\to\infty} \EE^{\QQ_k}[ \bsv\cdot \bB]
=\sup_{\QQ\in\MM_i} \EE^{\QQ}[\bsv\cdot\bB]
=\sup_{\bsp\in\mathcal{P}_i^{\bB}} \bsv\cdot \bsp.
\end{equation}
\end{lemma}
\begin{proof}
Our first  claim is  that
\begin{equation}\label{equ:lemma}
  \lim_{k\to\infty}
  \frac{\rho_i(-\boa_k\cdot\bB)}{\gamma_k}\geq S,\text{ where }S=\sup_{\QQ\in\MM_i}
  \EE^{\QQ}[\bsv\cdot \bB].
\end{equation}
For $\eps>0$ there exists $\QQ^{\eps}\in\MM_i$ such that
\[ \EE^{\QQ^{\eps}}[\bsv\cdot\bB]\geq S-\eps/2,\] and, for that choice
of $\QQ^{\eps}$, there exists $K\in\N$ such that for $k\geq K$ we have
\[ \tfrac{\alpha(\QQ^{\eps})+\delta}{\gamma_k}\leq \eps/2\text{ where } \delta=
\abs{\rho_i(-\boa_0\cdot\bB)}+\norm{\boa_0\cdot\bB}_{\linf}.\]
Thanks to convexity of $\rho_i$,  the ratio
  $\frac{\rho_i(-\boa_n\cdot\bB)-\rho_i(-\boa_0\cdot\bB)}{\gamma_n}$
  is  nondecreasing, so its limit as $n\to\infty$ exists in
  $(-\infty,\infty]$ and
\begin{equation}
\nonumber
   \begin{split}
     \lim_{k\to\infty} \frac{\rho_i(-\boa_k\cdot\bB)}{\gamma_k} &=
     \sup_{k\geq K}\frac{\rho_i(-\boa_k\cdot\bB)-\rho_i(-\boa_0\cdot\bB)}{\gamma_k}\\
     &= \sup_{k\geq K}\, \sup_{\QQ\in\MM_i}
     \left\{\frac{\EE^{\QQ}[\boa_k\cdot\bB]}{\gamma_k}-
       \frac{\alpha_i(\QQ)}{\gamma_k} - \frac{\rho_i(-\boa_0\cdot
         \bB)}{\gamma_k}
     \right\}\\
     &= \sup_{\QQ\in\MM_i}\sup_{k\geq K}
     \left\{\EE^{\QQ}[\bsv\cdot\bB]
       -\frac{\alpha_i(\QQ)+\rho_i(-\boa_0\cdot
         \bB)+\EE^{\QQ}[-\boa_0\cdot\bB] }{\gamma_k}
     \right\} \\
     &\geq \EE^{\QQ^{\eps}}[\bsv\cdot\bB]
     -\frac{\alpha_i(\QQ^{\eps})+\rho_i(-\boa_0\cdot
       \bB)+\EE^{\QQ^{\eps}}[-\boa_0\cdot\bB] }{\gamma_k} \geq
     S-\eps, \end{split}
\end{equation}
and (\ref{equ:lemma}) follows.

We continue by noting that since
$\QQ_k\in\partial\rho_i(-\boa_k\cdot\bB)$ and $\alpha_i(\cdot)\geq 0$, we have
\[ \tfrac{1}{\gamma_k} \rho_i(-\boa_k\cdot \bB)\leq
\tfrac{1}{\gamma_k}\EE^{\QQ_k}[\boa_k\cdot \bB]\leq
\EE^{\QQ_k}[\bsv\cdot\bB]+\tfrac{1}{\gamma_k} \norm{a_0\cdot\bB},\]
and so
\begin{equation}
   \nonumber
   \begin{split}
 \sup_{\QQ\in\MM_i} \EE^{\QQ_k}[\bsv\cdot\bB]&\leq \lim_{k\to\infty}
\tfrac{\rho_i(-\boa_k\cdot\bB)}{\gamma_k}\leq
\liminf_{k\to\infty}
\EE^{\QQ_k}[\bsv\cdot\bB]\leq \limsup_{k\to\infty}
\EE^{\QQ_k}[\bsv\cdot\bB]\\
&\leq
\sup_{\QQ\in\MM^{\bB}_i}\EE^{\QQ_k}[\bsv\cdot\bB].
   \end{split}
\end{equation}
We conclude the proof by noting that since $\MM_i^{\bB}\subseteq
\MM_i$, all the inequalities above are, in fact, equalities.
\end{proof}

\begin{lemma}\label{lem:gradient of rho}
  For $i\in\{1,2,...,I\}$, let $\bB\in (\linf)^n$ be a bundle of
  claims for which there is no $\boa\in\R^n$ such that
  $\boa\cdot\bB\in\ar_i$. If $\AA_i$ is weak-$\ast$ closed and
  strictly convex with respect to $\bB$, then the expectation
  $\EE^{\QQ}[\boa\cdot \bB]$ is the same for all
  $\QQ\in\partial\rho_i(\boa\cdot\bB)$,
the function $r_i:\R^n\to\R$,
  defined by
  \begin{equation}
   \label{equ:def-r-i}
   % \nonumber
   \begin{split}
     r_i(\boa)=\rho_i(\boa\cdot\bB)
   \end{split}
\end{equation}
is continuously differentiable and
\[\nabla r_i(\boa)=\EE^{\QQ}[-\bB],\text{ for any }
\QQ\in
\partial\rho_i(\boa\cdot\bB)
.\]
\end{lemma}

\begin{proof}
  Thanks to convexity of $r_i$ and Proposition I.5.3 in
  \cite{EkeTem99}, it will be enough to show that $\EE^{\QQ}[-\bB]$ is
  the unique subgradient of $r_i$ at $\boa$ for any
  $\QQ\in \partial\rho_i(\boa\cdot\bB)$.  To proceed, we suppose that
  $\boa^{\ast}\in\R^n$ satisfies $r_i(\bsdel)\geq
  r_i(\boa)-\boa^{\ast}\cdot(\boa-\bsdel)$, i.e.
\begin{equation}\label{equ:subgradient}
\rho_i(\boldsymbol{\delta}\cdot\bB)\geq\rho_i(\boa\cdot\bB)-\boa^{\ast}\cdot(\boa-\boldsymbol{\delta}),\text{
 }
\end{equation}
for all $\bsdel\in\R^n$, and that $\boa^{\ast}\neq \EE^{\QQ}[-\bB]$. 
Consider, first, the case when $-\boa^{\ast}\in\mathcal{P}_i(\bB)$, i.e., when there
exists $\hat{\boldsymbol{\delta}}\in\R^n$ such that
$\hat{\boldsymbol{\delta}}\neq\boa$ and
$\boa^{\ast}=\EE^{\hat{\QQ}}[-\bB]$ for some
$\hat{\QQ}\in\partial\rho_i(\hat{\bsdel})$. If we substitue
$\hat{\boldsymbol{\delta}}$ for ${\boldsymbol{\delta}}$ in
(\ref{equ:subgradient}) we get
$$\rho_i(\hat{\bsdel}\cdot\bB)+\EE^{\hat{\QQ}}[\hat{\bsdel}\cdot\bB]
\geq\rho_i(\boa\cdot\bB)+\EE^{\hat{\QQ}}[\boa\cdot\bB].$$ Note,
however, that
\begin{eqnarray*}
  \rho_i(\boa\cdot\bB) & \geq & \EE^{\hat{\QQ}}[-\boa\cdot\bB]-\alpha_i(\hat{\QQ}) =\EE^{\hat{\QQ}}[-\boa\cdot\bB]+\rho_i(\hat{\bsdel}\cdot\bB)+
  \EE^{\hat{\QQ}}[\hat{\boldsymbol{\delta}}\cdot\bB],
\end{eqnarray*}
and the equality holds if and only if
$\hat{\QQ}\in\partial\rho_i(\boa\cdot\bB)$.
This
implies that
$$\frac{1}{2}(\rho_i(\boa\cdot\bB)+\rho_i(\hat{\boldsymbol{\delta}}\cdot\bB))=
\rho_i \left(\frac{(\hat{\boldsymbol{\delta}}+\boa)\cdot\bB}{2}\right),$$
which contradicts the assumption of strict convexity with respect to $\bB$.

It is left to show that if
$\boa^{\ast}=(a^{\ast}_1,\dots,a^{\ast}_n)\in\R^n$ satisfies
(\ref{equ:subgradient}), then $-\boa^{\ast}\in\mathcal{P}_i(\bB)$. We
argue by contradiction and assume that this is not the case. By Lemma
\ref{lem:sup_sets} this means that there exists a component
$l\in\{1,2,...,n\}$ such that
either $a^{\ast}_l\geq\underset{\QQ\in\MM_i}{\sup}\{\EE^{\QQ}[-B_l]\}$ or
$a^{\ast}_l\leq\underset{\QQ\in\MM_i}{\inf}\{\EE^{\QQ}[-B_l]\}$.
We can assume without loss of generality that the former holds and that
$l=1$. The inequality (\ref{equ:subgradient}), in which
${\boa^{1}}=(a_1+1,a_2,a_3,...,a_n)$ is substituted for $\bsdel$,
implies that
$$\rho_i({\boa^{1}}\cdot\bB)+\underset{\QQ\in\MM_i}{\inf}
\EE^{\QQ}[B_1]\geq\rho_i(\boa\cdot\bB).$$  On the
other hand, for $\QQ^{1}\in\partial\rho_i(\boa^{1}\cdot\bB)$,
we have
\begin{eqnarray*}
  \rho_i(\boa\cdot\bB) & \geq & \EE^{\QQ^{1}}[-\boa\cdot\bB]-\alpha_i(\QQ^{1})
  =\rho_i({\boa^1}\cdot\bB)+\EE^{\QQ^{1}}[B_1]  >
  \rho_i({\boa^1}\cdot\bB)+\underset{\QQ\in\MM_i}{\inf}\{\EE^{\QQ}[B_1]\},
\end{eqnarray*}
a contradiction.
\end{proof}

\begin{remark}\label{rem:Z(p) is a function}
  If for $\bB\in (\linf)^n$ there is no $\boa\in\R^n$ such that
  $\boa\cdot\bB\in\ar_i$ and if $\AA_i$ is strictly convex with
  respect to $\bB$, the demand correspondence is $Z_i(\bsp)$ is
  non-empty only when $\bsp\in\sP_i(\bB)$.  Indeed,
  thanks to its definition as a minimizer of a differentiable convex
  function, the set $Z_i(\bsp)$, consists of the solutions $\boa$, of
  the equation
\begin{equation}
   \label{equ:qqq}
   \begin{split}
   \nabla r_i(\boa)=\bsp,
   \end{split}
\end{equation}
 where $r_i$ is defined in
  (\ref{equ:def-r-i}) above. Lemma \ref{lem:gradient of rho}
  states, however, that $\nabla r_i(\boa)\in\sP_i(\bB)$. Moreover,
  when $\bsp\in\sP_i(\bB)$, $Z_i(\bsp)$ is a singleton; this follows
  directly from strict convexity of $r_i$. Finally, it follows easily from
(\ref{equ:qqq}) that $Z_i$ is an injection on $\sP_i(\bB)$.
%   $-\partial\rho_i(\boa\cdot\bB)=\EE^{\QQ_i^{(\boa\cdot\bB)}}[\bB]=\bsp$.
%   For $\bsp\in\mathcal{P}_i(\bB)$, the map $Z_i(\bsp)$ is in fact a
%   function, i.e., $Z_i(\bsp)$ is a singleton. Indeed, if we assume
%   that for some price vector $\hat{\bsp}$, there exist
%   $\hat{\boldsymbol{\zeta}}$, $\tilde{\boldsymbol{\zeta}}\in\R^n$ such
%   that $\hat{\boldsymbol{\zeta}}\neq\tilde{\boldsymbol{\zeta}}$ and
%   $\EE^{\QQ^{(\hat{\boldsymbol{\zeta}}\cdot\bB)}}[\bB]=\EE^{\QQ^{(\tilde{\boldsymbol{\zeta}}\cdot\bB)}}[\bB]=\hat{\bsp}$,
%   then
% $$\rho_i(Z_i(\hat{\bsp})\cdot\bB)+Z_i(\hat{\bsp})\cdot\hat{\bsp}=\rho_i(\hat{\boldsymbol{\zeta}}\cdot\bB)+\hat{\boldsymbol{\zeta}}\cdot\hat{\bsp}=
% \rho_i(\tilde{\boldsymbol{\zeta}}\cdot\bB)+\tilde{\boldsymbol{\zeta}}\cdot\hat{\bsp}.$$
% Thanks to the strict convexity of $\rho_i$, for every $\lambda\in
% (0,1)$, we have that
% $$\rho_i((\lambda\hat{\boldsymbol{\zeta}}+(1-\lambda)\tilde{\boldsymbol{\zeta}})\cdot\bB)+
% (\lambda\hat{\boldsymbol{\zeta}}+(1-\lambda)\tilde{\boldsymbol{\zeta}})\cdot\hat{\bsp}<\rho_i(Z_i(\hat{\bsp})\cdot\bB)+Z_i(\hat{\bsp})\cdot\hat{\bsp}.$$
% The last inequality contradicts the fact that
% $$\rho_i(Z_i(\hat{\bsp})\cdot\bB)+Z_i(\hat{\bsp})\cdot\hat{\bsp}\leq
% \rho_i(\boa\cdot\bB)+\boa\cdot\bB,$$ which holds for every
% $\boa\in\R^n$.
\end{remark}

\begin{proposition}\label{pro:Z_i-is-decreasing}
Let $i\in\{1,2,...,I\}$ and $\bB\in (\linf)^n$ be a bundle
for which there is no $\boa\in\R^n$ such that
$\boa\cdot\bB\in\ar_i$. If $\AA_i$ is weak-$\ast$ closed and strictly
convex with respect to $\bB$, then the demand function $Z_i$ is
continuous in $\mathcal{P}_i(\bB)$ and satisfies the monotonicity property
$$(Z_i(\bsp_1)-Z_i(\bsp_2))\cdot(\bsp_1-\bsp_2)<0,$$ for every
$\bsp_1,\bsp_2\in\mathcal{P}_i(\bB)$ with $\bsp_1\neq\bsp_2$.
\end{proposition}
\begin{proof}
  The continuity is a direct application of the Berge's Maximum
  Theorem (see, for instance, Theorem 17.31 in \cite{AliBor06}). To
  establish monotonicity, we recall that the properties of $Z_i$ exposed
  in Remark \ref{rem:Z(p) is a function} above imply that
\begin{equation*}
\rho_i(Z_i(\bsp_j)\cdot\bB)+Z_i(\bsp_j)\cdot\bsp_j<\rho_i(\boa\cdot\bB)+\boa\cdot\bsp_j
\end{equation*}
for every $\boa\in\R^n$ with $\boa\neq Z_i(\bsp_j)$, for $j=1,2$.
Since $Z_i(\bsp_1)\neq Z_i(\bsp_2)$, we have \begin{equation*}
  \rho_i(Z_i(\bsp_1)\cdot\bB)-\rho_i(Z_i(\bsp_2)\cdot\bB)+(Z_i(\bsp_1)-Z_i(\bsp_2))\cdot\bsp_1<0
\end{equation*}
and
\begin{equation*}
  \rho_i(Z_i(\bsp_2)\cdot\bB)-\rho_i(Z_i(\bsp_1)\cdot\bB)-(Z_i(\bsp_1)-Z_i(\bsp_2))\cdot\bsp_2<0.
\end{equation*}
The required monotonicity follows when we add the above
inequalities. \end{proof}

\subsection{Equilibrium}
Our goal in this subsection is to prove existence and uniqueness of a
partial-equilibrium price for a given bundle $\bB\in(\linf)^n$.  We
follows the classical paradigm: the price vector $\bsp$ is a {\em
  partial-equilibrium price} of $\bB$, if, when $\bB$ trades at
$\bsp$, demand and supply for each of its components offset each
other, i.e., the market for $\bB$ clears. It should be noted that we
do not require that the agents' positions in the liquid markets clear
as well. A mathematically precise formulation of the above principle,
where we also consider the equilibrium allocation,
is given in the folowing definition:
\begin{definition}\label{def:gequilibrium}
We say that the pair
$(\bsp,\boa)\in\big(\bigcap_{i=1}^{I}\mathcal{P}_i(\bB)\big)\times
\mathbf{F}$ is a \textit{partial-equilibrium price-allocation
(PEPA)}\index{partial equilibrium price allocation}, if
$\boa_i=Z_i(\bsp)$\index{PEPA} for every $1\leq i\leq I$, i.e., if
\begin{equation}\label{equ:pepq}
    \sum_{i=1}^{I}Z_i(\bsp)=\bze\in\R^n.
\end{equation}
\end{definition}
The following two assumptions will be in effect throughout the section:
\begin{assumption}\label{ass:strictly convex}
The acceptance set $\AA_i$ is weak-$\ast$ closed and strictly convex
with respect to $\bB$, for all $i=1,2,...,I$.
\end{assumption}

\begin{assumption}\label{A1}
  For each $\boldsymbol{\delta}\in\R^n\setminus\{\mathbf{0}\}$
\[
\inf_{\QQ\in\MM} \EE^{\QQ}[\boldsymbol{\delta}\cdot\bB]
<
\sup_{\QQ\in\MM} \EE^{\QQ}[\boldsymbol{\delta}\cdot\bB].
\]
\end{assumption}

\begin{remark}
Assumption \ref{A1} implies, in particular, that for every
$i=1,2,...,I$, there is no
$\boldsymbol{\delta}\in\R^n\setminus\{\mathbf{0}\}$ such that
$\boldsymbol{\delta}\cdot\bB\in\mathcal{X}^{\infty}_i$. That
means, further, that there is no
$\boldsymbol{\delta}\in\R^n\setminus\{\mathbf{0}\}$ such
$\boldsymbol{\delta}\cdot\bB\in\ar_i$.
\end{remark}

\begin{remark}
  If Assumptions \ref{ass:strictly convex} and \ref{A1} hold true and
  the agents are in a Pareto optimal configuration then the pair
  $(\bsp,\boa)\in\big(\bigcap_{i=1}^{I}\mathcal{P}_i(\bB)\big)\times
  \mathbf{F}$ is a PEPA if and only if $\boa=\bze$. To see this,
  assume that there exists an
  $\boa\in\mathbf{F}\backslash\{\mathbf{0}\}$ such that $(\bsp,\boa)$
  is a PEPA. Then, by the definition of $Z_i$, we have
  $\rho_i(\boa_i\cdot\bB)+\boa_i\cdot\bsp\leq 0$ for all $i$, and, due
  to the strict convexity of the risk measures, the inequality is
  strict for at least two agents. This implies that $\sum_{i=1}^I
  \rho_i(\boa_i\cdot\bB)<0$, which contradicts the assumption of
  Pareto optimality. The economic interpretation of the above
  statement is clear - the agents will not engage in trade if they are
  already in a configuration which cannot be improved upon.
\end{remark}
The following theorem contains the main result of the section:
\begin{theorem}\label{thm:PEPA}
  Under Assumptions \ref{ass:not empty intersectio}, \ref{ass:strictly
    convex} and \ref{A1}, there exists a unique PEPA
  $(\hat{\bsp},\hat{\boa})\in\big(\bigcap_{i=1}^I\mathcal{P}_i(\bB)\big)\times\mathbf{F}$.
\end{theorem}

\begin{proof}
We first define the strictly convex function $f:\R^{(I-1)\times
n}\to\R$ by
\begin{equation}\label{equ:function_f}
    f(\boa)=\rho_1(\boa_1\cdot \bB)+\rho_2(\boa_2\cdot \bB)+...+\rho_{I-1}(\boa_{I-1}\cdot \bB)+\rho_{I}((-\sum_{i=1}^{I-1}\boa_i)\cdot \bB).
\end{equation}
If for some $\tilde{\boa}=(\tilde{\boa}_1,\dots,
\tilde{\boa}_{I-1})\in\R^{(I-1)\times n}$ we have $\nabla
f(\tilde{\boa})=\mathbf{0}$, then $\tilde{\boa}$ is the unique
minimizer of $f$. Moreover, such $\tilde{\boa}$
 satisfies
$\nabla\rho_i(\tilde{\boa}_i\cdot\bB)=\nabla\rho_I(-\left(\sum_{i=1}^{I-1}\tilde{\boa}_i\right)\cdot\bB)$,
for every $i=1,2,...,I-1$.
The latter means that for any
$\QQ_i\in\partial\rho_i(\tilde{\boa}_i\cdot\bB)$, $1\leq i \leq I-1$ and any
$\QQ_I\in\partial\rho_i(-(\sum_{i=1}^{I-1}\tilde{\boa}_i)\cdot\bB)$, we have
$$\EE^{\QQ_i}[\bB]=\EE^{\QQ_I}[\bB].$$
Therefore, the price vector $\hat{\bsp}=\EE^{\QQ_i}[\bB]$ satisfies
$Z_i(\hat{\bsp})=\tilde{\boa}_i$ for every $i=1,2,...,I-1$ and
$Z_I(\hat{\bsp})=-\sum_{i=1}^{I-1}\tilde{\boa}_i$.  In other words, if
$\hat{\boa}$ denotes the allocation whose rows are given by
$\hat{\boa}_i=\tilde{\boa}_i$, for $i=1,2,...,I-1$ and
$\hat{\boa}_I=-\sum_{i=1}^{I-1}\tilde{\boa}_i$, the pair
$(\hat{\bsp},\hat{\boa})$ is a partial equilibrium price
allocation. In fact, it is the unique one, since if we assume the
existence of another PEPA
$(\check{\bsp},\check{\boa})\neq(\hat{\bsp},\hat{\boa})$, we get that
$\check{\bsp}=\EE^{\QQ}[\bB]$, for any
$\QQ\in\partial\rho_i(\check{\boa}_i\cdot\bB)$, which, in turn,
implies that $\nabla f(\check{\boa})=\mathbf{0}$. The latter equation
contradicts the uniqueness of the minimizer of the function $f$.

We are left with the task of showing that $\nabla f(\boa)$ has a root,
and assume, per contra, that this is not the case. Then, by the
continuity of $f$, we deduce that for, each $m\in\mathbb{N}$, there
exists $\boa^{(m)}\in \mathbf{D}_m=\{\boa\in\R^{(I-1)\times
  n}:\norm{\boa}_1=\sum_{k=1}^{(I-1)}\sum_{j=1}^n |a_{k,j}| \leq m\}$ such that $f(\boa^{(m)})\leq
f(\boa)$ for all $\boa\in\mathbf{D}_m $. Furthermore, by the strict
convexity of $f$, it follows that $||\boa^{(m)}||_1=m$. Hence, thanks
to the results of, e.g., Chapter 1 in \cite{BorLew00},
a contradiction would be reached if the following coercivity condition held:
\begin{equation}\label{equ:contradiction}
    F=\underset{m\to\infty}{\liminf}\frac{f(\boa^{(m)})}{m}>0.
\end{equation}
By passing to a subsequence (if necessary), we can assume without loss
of generality that the limits
$F=\underset{k\to\infty}{\lim}\frac{f(\boa^{(k)})}{k}\in\R$ and
$\boa^{(0)}_i=\lim \tfrac{\boa^{(k)}_i}{k}\in\R^n$, $i=1,2,\dots,I-1$ exist.
Since
\[
\abs{\frac{\rho_i(\boa_i^{(k)}\cdot\bB)}{k}- \frac{\rho_i(k
    \boa_i^{(0)}\cdot\bB)}{k}}\leq
\norm{\tfrac{\boa_i^{(k)}}{k}-\boa_i^{(0)}}\ \norm{\bB}_{(\linf)^n}\to
0,\] Lemma \ref{lem:sup_sets} implies that
\begin{equation}
\nonumber
   \begin{split}
     \underset{k\to\infty}{\lim}\frac{\rho_i(\boa^{(k)}_i\cdot
       \bB)}{k}&=\sup_{\QQ\in\MM_i} \EE^{\QQ}[\boa_i^{(0)}\cdot\bB],\text{ for $1\leq i\leq I-1$, and }\\
     \underset{k\to\infty}{\lim}\frac{\rho_{I}( -\sum_{j=1}^{I-1}
       \boa^{(k)}_j\cdot \bB)}{k}&= \sup_{\QQ\in\MM_I}
     \EE^{\QQ}[-\sum_{j=1}^{I-1} \boa_j^{(0)}\cdot\bB].
   \end{split}
\end{equation}
Consequently, (\ref{equ:contradiction}) follows from
\begin{equation}
 \nonumber
   \begin{split}
 F&=
\sum_{j=1}^{I-1} \sup_{\QQ\in\MM_j} \EE^{\QQ}[\boa_j^{(0)}\cdot\bB]
+\sup_{\QQ\in\MM_I} \EE^{\QQ}[-\sum_{j=1}^{I-1} \boa_j^{(0)}\cdot\bB]
\\ &\geq
\sup_{\QQ\in\MM} \EE^{\QQ}[ \sum_{j=1}^{I-1} \boa_j^{(0)}\cdot\bB]
-\inf_{\QQ\in\MM} \EE^{\QQ}[\sum_{j=1}^{I-1} \boa_j^{(0)}\cdot\bB]>0,
   \end{split}
\end{equation}
where the strictness of the last inequality follows from
Assumption \ref{A1}.
\end{proof}
% \begin{remark}
% By the uniqueness of the minimizer $\hat{\boa}\in\mathbf{F}$ of
% the function $f$, we get that any agent $i$ who participates in
% the equilibrium (that is $\hat{\boa}_i\neq 0$) enjoys a risk
% reduction (improvement), i.e,
% $\rho_i(\hat{\boa}_i\cdot\bB-\hat{\boa}_i\cdot\hat{\bsp})<0$, where
% $\hat{\bsp}$ is the partial equilibrium price (PEP).
% \end{remark}
\begin{remark}
  It follows from Theorem \ref{thm:PEPA} that the PEPA corresponding
  to a bundle $\bB$ is of the form $(\hat{\bsp},\boa)$ with $\boa\neq
  0$ {\em if and only if} the agents are \underline{not} in a {\em
    constrained Pareto-optimal configuration}; meaning that the
    marginal prices $\EE^{\QQ^i}[\bB]$ are not all equal.  A simple
    consequence of this statement is that, in that case,
    $\agset^{\bB}\neq\set{0}$.
\end{remark}

\bigskip

\section{The well-posedness of the equilibrium pricing}\label{sec:stability}

The exact shape of agents' acceptance sets, which incorporate their
risk preferences, endowments and investment goals, is extremely
difficult to estimate in practice.  It is therefore natural to ask
whether the induced equilibrium pricing is stable with respect small
perturbation in the agents' acceptance sets. To be more precise, we
want to check whether the equilibrium pricing scheme, presented in
section \ref{sec:PEPA}, is a \textit{well-posed problem} in the sense
of Hadamard (see \cite{Had02}), i.e., if its solution exists, is
unique and stable with respect to the input data (the agents'
acceptance sets in this case).  Having solved the problem of existence
and uniqueness (see Theorem \ref{thm:PEPA}), we turn our attention
to the following question: can we specify a convergence (concept)
$\overset{\circledast}{\longrightarrow}$ for $I$-tuples of
the weak-$\ast$ closed acceptance sets $\left(
  \AA_i^{(m)}\right)_{i=1}^I=\left(\AA_1^{(m)},\AA_2^{(m)},...,\AA_I^{(m)}\right)$,
 for which
\begin{equation}\label{question}
\left(\AA^{(m)}_1,\AA^{(m)}_2,...,
\AA^{(m)}_I\right)\overset{\circledast}{\longrightarrow}
\left(\AA_1,\AA_2,...,\AA_I\right)\Longrightarrow
    \left(\hat{\bsp}^{(m)},\hat{\boa}^{(m)}\right)\to(\hat{\bsp},\hat{\boa}),
\end{equation}
where $(\hat{\bsp}^{(m)},\hat{\boa}^{(m)})$ is the PEPA obtained by
the acceptance sets $\left(\AA_i^{(m)}\right)_{i=1}^I$ and
$(\hat{\bsp},\hat{\boa})$ is the corresponding to
$\left(\AA_i\right)_{i=1}^I$ PEPA?

As we shall explain shortly, it turns out that the right notion is related to
 \textit{Kuratowski} convergence
(see Chapter 8 in \cite{Luc06} and Chapter 7 in \cite{RocWet98} for a
further analysis):
\begin{definition}\label{def:Kur-conve}
A sequence of closed subsets $C_m\subseteq\R^l$, $l\in\mathbb{N}$, converges to the subset $C$ in
Kuratowski sense (and we write $C_m\overset{K}{\longrightarrow}C$) if
\begin{equation}\label{equ:Kur-conve}
  \text{Ls }C_m\subseteq C\subseteq \text{Li }C_m,
\end{equation}
where $$\text{Li }C_m=\sets{c\in\R^l}{c=\lim c_k, c_k\in C_k \text{
    eventually}}$$ and
$$\text{Ls }C_m=\sets{c\in\R^l}{c=\lim c_k,
  c_k\in C_{n_k}, n_k \text{ a subsequence of integers}}.$$
\end{definition}
We say that a sequence $\{f_m\}_{m\in\N}$ of lower semi-continuous functions
$f_m:\R^l\to\R$ converges to a function $f$ in the Kuratowski sense (and
we write $f_m\overset{K}{\longrightarrow}f$) if
$\epi(f_m)\overset{K}{\longrightarrow}\epi(f)$. We remind the reader
that the epigraph of a function $f:\R^n\to\R$ is the set
$\epi(f)=\sets{(\boa,c)\in\R^n\times\R}{f(\boa)\leq c}$.
 A characterization of
the Kuratowski convergence for sequences of functions is given by
Theorem 8.6.3 in \cite{Luc06} (see also Proposition 7.2 in
\cite{RocWet98}); $f_m\overset{K}{\longrightarrow}f$ if and only if
the following two conditions hold:
\begin{itemize}
\item [(a)] For every $x\in\R^k$ and every sequence $x_n$ such that
  $x_n\to x$, $\liminf f_n(x_n)\geq f(x)$ and
\item [(b)] For every $x\in\R^k$ there exists a sequence $x_n$ such
  that $x_n\to x$ and $\limsup f_n(x_n)\leq f(x)$.
\end{itemize}
The Kuratowski convergence and its versions for more general
topological spaces have been extensively used in the study of the
well-posedness of a variety of variational problems (see
\cite{RocWet98} for problems in $\R^n$ and \cite{DonZol93} and
\cite{Luc06} for general spaces).

In what follows, for each agent $i$, we consider a sequence of
weak-$\ast$ closed acceptance sets $\AA_i^{(m)}$ and a limiting
weak-$\ast$ closed acceptance set $\AA_i$, all of which satisfy the
axioms \textit{Ax1}-\textit{Ax4}.  The induced convex capital
requirements are denoted by $\rho_i^{(m)}(\cdot)$ and $\rho_i(\cdot)$
respectively, and $\MM_i^{(m)}$ and $\MM_i$ stand for the effective
domains of the corresponding penalty functions, $\alpha_i^{(m)}$ and
$\alpha_i$.

In the effort to show that the Kuratowski convergence allows for a positive answer to our central question, we establish the following auxiliary result:
\begin{lemma}\label{lem:rho is strict convex}
  For a given bundle of claims $\bB$, if
  $\AA_i^{(m)}(\bB)\overset{K}{\longrightarrow}\AA_i(\bB)$ for every
  $i\in\{1,2,...,I\}$, then the sequence of functions
  $\boa\ni\R^n\mapsto\rho_{i}^{(m)}(\boa\cdot\bB)$ converges
  point-wise to the function
  $\boa\ni\R^n\mapsto\rho_{i}(\boa\cdot\bB)$.
\end{lemma}
\begin{proof}
By \eqref{risk measure definition}, for a bundle $\bB\in(\linf)^n$,
the set $\AA(\bB)$
is the \textit{epigraph} of the function
$\R^n\ni\boa\mapsto\ra(\boa\cdot\bB)\in\R$.
Hence, $\AA_i^{(m)}(\bB)\overset{K}{\longrightarrow}\AA_i(\bB)$
  is equivalent to the Kuratowski convergence of the sequence of
  functions $\boa\ni\R^n\mapsto\rho_{i}^{(m)}(\boa\cdot\bB)$.

  It is shown in \cite{RocWet98}, Theorem 7.17, that for any sequence
  $\left(f_m\right)_{m\in\mathbb{N}}$ of convex functions on $\R^n$,
  $f_m\overset{K}{\longrightarrow}f$ implies that $f_m\rightarrow f$
  point-wise in $R^n$, provided that $f$ is a convex, lower
  semi-continuous function  and its effective domain has
  non-empty interior.
 It is, therefore,  enough to observe that the function
  $\boa\ni\R^n\mapsto\rho_{i}(\boa\cdot\bB)$ is convex and lower
  semi-continuous in $\R^n$, since $\rho_i:\linf\to\R$ is convex and
  $\sigma(\linf,\lone)$-lower semi-continuous risk measure.
\end{proof}
As the reader can easily check, Kuratowski convergence will not, in
general, preserve strict convexity. In order to guarantee that the
limiting acceptance set $\AA_i$ is strictly convex with respect to the
fixed bundle of claims $\bB$, we must assume that the strict convexity
of $\AA_i^{(m)}$ with respect to $\bB$ satisfies a certain uniformity
criterion.
\begin{definition}\label{def:uni.str.convex}
  A sequence of acceptance sets $(\AA^{(m)})_{m\in\mathbb{N}}$ is
  \textit{uniformly strictly convex} with respect to a bundle
  $\bB\in(\linf)^n$, if for every
  $(\boa,c),(\boldsymbol{\delta},k)\in\AA^{(m)}(\bB)$ such that
  $\boa\neq\boldsymbol{\delta}$, the following statement holds:

  for every $\lambda\in (0,1)$ there exists a random variable
  $E\in\linf_+$, such that $\QQ[E>0]>0$, for some $\QQ\in\partial
  \rho_i^{(m)}{((\lambda\boa+(1-\lambda)\boldsymbol{\delta})\cdot\bB)}$
  and $$\lambda(\boa\cdot\bB+c)+
  (1-\lambda)(\boldsymbol{\delta}\cdot\bB+k)-E\in\AA^{(m)},$$
for all $m\in\N$.
\end{definition}

  It follows from the definition of Kuratowski convergence, that if
  $(\AA_i^{(m)})_{m\in\mathbb{N}}$ is uniformly strictly convex with
  respect to $\bB$ and
  $\AA_i^{(m)}(\bB)\overset{K}{\longrightarrow}\AA(\bB)$, then $\AA_i$ is
  also strictly convex with respect to $\bB$.  This fact and
  Lemma \ref{lem:gradient of rho} imply, in particular, that the
  function $\boa\ni\R^n\mapsto\rho_{i}(\boa\cdot\bB)$ is strictly
  convex and  differentiable on $\R^n$, if we further assume
  that $\bB$ is not redundant, i.e., there is no
  $\boldsymbol{\delta}\in\R^n$, such that
  $\boldsymbol{\delta}\cdot\bB\in\ar_i$.

\begin{assumption}\label{ass:uniform convexity}
The sequence  $\{\AA_i^{(m)}\}_{m\in\N}$ of acceptance sets is uniformly
strictly convex with respect to the bundle $\bB$.
\end{assumption}

\begin{assumption}\label{ass:stability}
 $\emptyset\neq \bigcap_{i=1}^I\MM_i^{(m)}\subseteq\MM$, for all $m\in\N$.
\end{assumption}

\begin{assumption}\label{ass:PEPA for each m}
For each $m\in\N$ and $\boldsymbol{\delta}\in\R^n\setminus\{\mathbf{0}\}$
\[
\inf_{\QQ\in\MM^{(m)}} \EE^{\QQ}[\boldsymbol{\delta}\cdot\bB]
<
\sup_{\QQ\in\MM^{(m)}} \EE^{\QQ}[\boldsymbol{\delta}\cdot\bB].
\]
\end{assumption}

It follows from Theorem \ref{thm:PEPA} that under the Assumptions
\ref{ass:uniform convexity}, \ref{ass:stability} and \ref{ass:PEPA for
  each m} there exists a unique PEPA,
$(\hat{\bsp}^{(m)},\hat{\boa}^{(m)})$, for every
$m\in\mathbb{N}$. Furthermore, the induced strict convexity of $\AA_i$
with respect to $\bB$ means that the conditions for existence and
uniqueness of PEPA hold even for the limiting risk measures
$\rho_i$. Moreover, it turns out that those same conditions guarantee
that the problem is well posed:
\begin{theorem}\label{thm:stability}
  Under Assumptions \ref{ass:uniform convexity}, \ref{ass:stability} and \ref{ass:PEPA for each m},
  for each $m\in\mathbb{N}$ there exists
  a unique PEPA $(\hat{\bsp}^{(m)},\hat{\boa}^{(m)})$ for agents with
  acceptance sets $\left(\AA_i^{(m)}\right)_{i=1}^I$. Also, the
  convergence $$\AA_i^{(m)}(\bB)\overset{K}{\longrightarrow}\AA_i(\bB)$$
  for every $i\in\{1,2,...,I\}$ implies that
\begin{itemize}
\item [(i)] There exists a unique PEPA $(\hat{\bsp},\hat{\boa})$ for
  agents with acceptance sets $\left(\AA_i\right)_{i=1}^I$ and
\item [(ii)]
  $(\hat{\bsp}^{(m)},\hat{\boa}^{(m)})\longrightarrow(\hat{\bsp},\hat{\boa})$
  in $\R^n\times\R^{n\times I}$.
\end{itemize}
\end{theorem}

\begin{proof}
  The existence and the uniqueness of the PEPA for agents with
  acceptance sets $\left(\AA_i^{(m)}\right)_{i=1}^I$ follows directly
  from Theorem \ref{thm:PEPA}.  By Lemma \ref{lem:rho is strict
    convex}, the Kuratowski convergence
  $\AA_i^{(m)}(\bB)\overset{K}{\longrightarrow}\AA_i(\bB)$ implies
  that $\rho_i^{(m)}(\boa\cdot\bB)\to\rho_i(\boa\cdot\bB)$, for every
  $\boa\in\R^n$ and that the function
  $\boa\ni\R^n\mapsto\rho_i(\boa\cdot\bB)$ is strictly convex. Then,
  the existence and the uniqueness of the PEPA
  $(\hat{\bsp},\hat{\boa})$, for agents with acceptance sets
  $\left(\AA_i\right)_{i=1}^I$ is guaranteed again by Theorem
  \ref{thm:PEPA}.

  Following the lines of the proof of Theorem \ref{thm:PEPA}, for each
  $m\in\mathbb{N}$ we define the strictly-convex function
  $f_m:\R^{n\times(I-1)}\to\R$ by
  \begin{equation}\label{equ:function f_m}
    f_m(\boa)=\rho^{(m)}_1(\boa_1\cdot \bB)+\rho^{(m)}_2(\boa_2\cdot \bB)+...+\rho^{(m)}_{I-1}(\boa_{I-1}\cdot \bB)+\rho^{(m)}_{I}((-\sum_{i=1}^{I-1}\boa_i)\cdot \bB),
\end{equation}
and note that it admits a unique minimizer,
$\tilde{\boa}^{(m)}\in\R^{n\times(I-1)}$ (where in fact,
$\tilde{\boa}_i^{(m)}=\hat{\boa}_i^{(m)}$ for every
$i=1,2,...,I-1$). Similarly, we define the function
\begin{equation}\label{equ:function f_m}
    f(\boa)=\rho_1(\boa_1\cdot \bB)+\rho_2(\boa_2\cdot \bB)+...+\rho_{I-1}(\boa_{I-1}\cdot \bB)+\rho_{I}((-\sum_{i=1}^{I-1}\boa_i)\cdot \bB).
\end{equation}
which is also strictly convex and has a unique minimizer
$\tilde{\boa}\in\R^{n\times(I-1)}$ (where,
$\tilde{\boa}_i=\hat{\boa}_i$ for every $i=1,2,...,I-1$). Note that
the point-wise convergence
$\rho_i^{(m)}(\boa\cdot\bB)\to\rho_i(\boa\cdot\bB)$, trivially implies
that $f_m(\boa)\to f_m(\boa)$, for every $\boa\in\R^{n\times(I-1)}$.
In order to show that $\hat{\boa}^{(m)}\to\hat{\boa}$, as
$m\to\infty$, we first recall a well-known result (see for instance
Example I.7 in \cite{DonZol93}) that if $f$ is a convex,
lower-semicontinuous function and has a minimizer $\tilde{\boa}$, then
for every sequence $\boldsymbol{\delta}^{(m)}\in\R^{n\times(I-1)}$
such that
$$f(\boldsymbol{\delta}^{(m)})\to f(\tilde{\boa}),$$
it holds that $\boldsymbol{\delta}^{(m)}\to \tilde{\boa}$. In other
words, the problem of minimizing $f$ in $\R^{n\times(I-1)}$ is
\textit{well posed in the sense of Tykhonov}. This implies that for
every $\varepsilon>0$ there exists $b\in\R_+$ such that
\begin{equation}\label{equ:stability2}
  \sets{\boa\in\R^{n\times(I-1)}}{f(\boa)\leq b+f(\tilde{\boa})}
  \subseteq\sets{\boa\in\R^{n\times(I-1)}}{||\boa-\tilde{\boa}||<\varepsilon}.
\end{equation}
% Indeed, if it were not the case, there would exist $\varepsilon>0$ and
% a sequence $\breve{\boa}^{(m)}$ such that
% $$f(\breve{\boa}^{(m)})\leq \frac{1}{m}+f(\tilde{\boa})
% \text{ and }||\breve{\boa}^{(m)}-\tilde{\boa}||>\varepsilon.$$ That
% would mean that $f(\breve{\boa}^{(m)})\to f(\tilde{\boa})$ but
% $\breve{\boa}^{(m)}\nrightarrow \tilde{\boa}$ - a contradiction with
% Tykhonov well-posedness.

By Lemma II.21 in \cite{DonZol93}, for every $b\in\R_+$ and
 sufficiently large $m$ we have
\begin{equation}\label{equ:stability1}
  \sets{\boa\in\R^{n\times(I-1)}}{f_m(\boa)\leq b+
    f_m(\tilde{\boa}^{(m)})}\subseteq
  \sets{\boa\in\R^{n\times(I-1)}}{f(\boa)\leq 2b+f(\tilde{\boa})}.
\end{equation}
Combination of \eqref{equ:stability1} and \eqref{equ:stability2}
yields the convergence $\tilde{\boa}^{(m)}\to\tilde{\boa}$, which
trivially implies the convergence of partial equilibrium allocations,
$\hat{\boa}^{(m)}\to\hat{\boa}$.   The definition of the
equilibrium price yields that
$$\nabla\rho_i^{(m)}(\hat{\boa}_i^{(m)}\cdot\bB)=-\hat{\bsp}^{(m)},$$
for every $m\in\mathbb{N}$.  Theorem 25.7 in \cite{Roc70}
implies that the convergence
$\rho_i^{(m)}(\boa\cdot\bB)\to\rho_i(\boa\cdot\bB)$ for every
$\boa\in\R^n$ and the fact that the limiting function
$\boa\ni\R^n\mapsto\rho_i(\boa\cdot\bB)$ is differentiable in $\R^{n}$
yield that
$$ \nabla\rho_i^{(m)}(\boa\cdot\bB)\to\nabla\rho_i(\boa\cdot\bB),$$
for every $\boa\in\R^n$ and every $i=\{1,2,...,I\}$. Furthermore, the same Theorem states that this
convergence is uniform on compacts in
$\R^n$, so
$$\hat{\bsp}^{(m)}=\nabla\rho_i^{(m)}(\hat{\boa}_i^{(m)}\cdot\bB)
\longrightarrow \nabla\rho_i(\hat{\boa}_i\cdot\bB)=\hat{\bsp}.$$
\end{proof}
We conclude with an example in which we show what Kuratowski convergence
looks like in a familiar setting:
\begin{example}
  We consider the utility-based acceptance sets discussed in Example
  \ref{exp:utility} for the agent $i$ (see \cite{HugKra04} for
  technical details) and we consider a sequence of utility functions
  $\left(U^{(m)}_i\right)_{m\in\mathbb{N}}$, a sequence of probability
  measures $\left(\PP^{(m)}\right)_{m\in\mathbb{N}}$, and a sequence
  of initial wealths $\left(x^{(m)}_i\right)_{m\in\mathbb{N}}$.
  For every $B\in\linf$, $x\in\R_+$ and $m\in\mathbb{N}$, we define the
  indirect utility
$$u^{(m)}_i(x|B)=\underset{X\in\mathcal{X}}{\sup}\EE^{\PP^{(m)}}[U_i^{(m)}(X+B)],$$
where $\mathcal{X}$ is a set of admissible strategies (see page 848 in
\cite{HugKra04} for the exact definition).  The corresponding sequence
of acceptance sets is then given by
$$\AA_i^{(m)}=\sets{B\in\linf}{u^{(m)}_i(x_i^{(m)}|B)\geq u^{(m)}_i(x_i^{(m)}|0)}.$$
for every $m\in\mathbb{N}$.  It was proved in \cite{KarZit07}, Theorem
1.5, that the following convergence conditions
$$\PP^{(m)}\to\PP \text{ in total variation, }
U^{(m)}_i\to U_i \text{ point-wise in }\R_+ \text{ and }x^{(m)}_i\to
x_i,$$ (together with some additional technical assumptions), yield
that for every non-redundant bundle $\bB\in(\linf)^n$, we have that
\begin{equation}\label{equ:stability exp}
    u^{(m)}_i(x_i^{(m)}|\boa^{(m)}\cdot\bB)\to u_i(x_i|\boa\cdot\bB),
\end{equation}
for every sequence $\boa^{(m)}\in\R^n$ that converges to some
$\boa\in\R^n$.  It is, then, straightforward to get that
\eqref{equ:stability exp} imply that
$\AA_i^{(m)}(\bB)\overset{K}{\longrightarrow}\AA_i(\bB)$, which in
turn guarantees that the equilibrium price-allocation of $\bB$ is
well-posed.
\end{example}

\bigskip

%\bibliographystyle{mybibstyle}
%\def\cprime{$'$} \def\cprime{$'$}

%\bibliography{./local}

\end{document}